\documentclass[pdflatex,sn-mathphys-num]{sn-jnl}
\usepackage{graphicx}%
\usepackage{multirow}%
\usepackage{mathrsfs}%
\usepackage[title]{appendix}%
\usepackage{xcolor}%
\usepackage{textcomp}%
\usepackage{manyfoot}%
\usepackage{booktabs}%
\usepackage{algorithm}%
\usepackage{algorithmicx}%
\usepackage{algpseudocode}%
\usepackage{listings}%
\theoremstyle{thmstyleone}%
\newtheorem{theorem}{Theorem}%
\newtheorem{proposition}[theorem]{Proposition}%
\theoremstyle{thmstyletwo}%
\newtheorem{remark}{Remark}%
\theoremstyle{thmstylethree}%
\usepackage{amsmath,amsthm,amsfonts,amssymb,amscd}
\usepackage{mathtools}
\usepackage[thinc]{esdiff}

\usepackage{graphicx}
\usepackage{a4wide}
\usepackage{hyperref}
\usepackage{lipsum}
\usepackage{fancyhdr}
\usepackage{xcolor}
\usepackage{subfigure}
\makeatletter
\def\@seccntformat#1{\@ifundefined{#1@cntformat}%
   {\csname the#1\endcsname\quad}  
   {\csname #1@cntformat\endcsname}
}
\let\oldappendix\appendix 
\renewcommand\appendix{%
    \oldappendix
    \newcommand{\section@cntformat}{\appendixname~\thesection\quad}
}
\makeatother

\begin{document}
	
	\title[Nonlinear Feedbacks Between Host Behavior and Vector Adaptation in a Multi-Host Vector-Borne Disease Model]{Nonlinear Feedbacks Between Host Behavior and Vector Adaptation in a Multi-Host Vector-Borne Disease Model}
	\author[1]{\fnm{Shravani} \sur{Shetgaonkar}}\email{shetgaonkarshravani@gmail.com}
	
	\author*[]{\fnm{Anupama} \sur{Sharma}$^*$}\email{anupamas@goa.bits-pilani.ac.in}
	
	\affil[]{\orgdiv{Department of Mathematics}, \orgname{Birla Institute of Technology and Science, Pilani, K K Birla Goa Campus}, \orgaddress{\street{Zuarinagar}, \city{Sancoale}, \postcode{403726}, \state{Goa}, \country{India}}}

\abstract{
Insecticide-treated nets (ITN) are an effective and low-cost intervention for controlling vector-borne disease (VBD), however, their use depends on individual decisions based on perceived cost and risk of infection. This study investigates a nonlinear multi-host model for the transmission of VBD with endogenous strategic control. We assume that hosts' adoption of ITN emerges from the payoff-based decision-making, creating a nonlinear coupling with disease prevalence. We model vector preference as a function of ITN coverage to probe the complex interplay among individual choices, disease prevalence, and its control in a multi-host setting. The qualitative behavior of the system is characterized by the thresholds $R_0$ and $R_c$, which determine the existence and local stability of the disease-free and endemic equilibria. The system exhibits rich dynamical behavior; hence, we provide a bifurcation analysis identifying the conditions for saddle-node and Hopf bifurcations. Our results demonstrate that the interaction between the perceived cost of ITN and the infection risk can induce critical transitions, including regime shift from stable endemic states to sustained periodic oscillations. Furthermore, we identify a counterintuitive effect whereby complete ITN adoption by the primary host can increase the overall prevalence in the secondary host due to adaptive shifts of vector feeding behavior.}
	
	\maketitle
	
		\section{Introduction}\label{sec1}
	
Vector-borne disease (VBD) dynamics are shaped not only by epidemiological processes but also by human behavior and vector biting patterns. Insecticide-treated nets (ITNs) are among the most effective interventions for reducing their transmission~\cite{mukhtar2018modelling}. However, the effectiveness of ITN depends on human compliance, which is shaped by discomfort, economic costs, and risk perception~\cite{oraby2025insecticide}. When ITN coverage is high, VBD prevalence declines, because of which individuals may perceive a lower infection risk and discontinue its use. This reduced ITN use can increase the chances of disease resurgence  \cite{agusto2013impact}. Hence, ITN usage evolves dynamically through feedback between disease prevalence, perceived risk, and individual behavioral adaptation. Moreover, ITN usage alters host accessibility to vectors and can redistribute vector bites in a multi-host setting~\cite{le2007elaborated,stone2018evolution}. This raises the question of whether the coupled interaction between adaptive protective behaviour and vector biting behaviour can generate emergent dynamical phenomena such as oscillatory outbreaks, bistability, or regime shift in multi-host disease systems.

Evolutionary game theory provides a natural framework for modelling adaptive protection behavior in response to the disease threat~\cite{bauch2005imitation,li2025effect,sarkar2025modeling,zhang2026analysis}. Previous studies \cite{honjo2014n,broom2016game,han2020voluntary,fortunato2021mathematical,laxmi2022evolutionary,rychtavr2022game,angina2022game,davis2024mathematical} have used the game theoretic framework to primarily focus on the outcome of individuals weighing the risk of infection against the perceived cost of ITN use. Individuals often choose between two strategies to protect or not to protect themselves with bed net usage, based on the perceived cost associated with each choice~\cite{onifade2024dynamic}. When ITN is not used by a fraction of the population, the vector bites can get concentrated on unprotected individuals, potentially increasing their infection risk \cite{miller2016risk,birget2015epidemiological,thomsen2017mosquito,enahoro2020long}. However, Demers \textit{et al.} \cite{demerspp2018} argued that protection status is not static in reality, as individuals may discontinue or re-adopt preventive measures in response to changing disease risk and perceived costs. While these studies demonstrate that both vector biting redistribution and adaptive protection behavior can independently influence disease dynamics, their coupled interaction remains poorly understood, particularly in multi-host systems where vectors may shift feeding toward alternative hosts.  

In multi-host systems, ITN use by the preferred host can increase the vector biting preference for alternative hosts, reshaping long-term disease transmission dynamics~\cite{stone2018evolution}. Moreover, protective measures adopted by the preferred host can trigger adaptive shifts in vector preferences and increase the overall prevalence~\cite{shravani&sharma}. Hence, the effects of adaptive protection cannot be understood independently of the changes in vector biting patterns induced by heterogeneous host protection. Changes in vector biting rates and density can affect individuals' strategic choices, generating a coupled feedback between adaptive bed net use, host exposure, and disease prevalence, thereby giving rise to rich nonlinear dynamics~\cite{weitz2016oscillating,hauert2008ecological}. Therefore, the collective dynamics of the system may differ substantially from those predicted by static transmission models, motivating our work to study the feedback between adaptive ITN use and vector biting preference that shapes disease dynamics. 

In this work, we propose a multi-host VBD model coupling adaptive ITN use with vector biting preference through a co-evolutionary game-theoretic framework. Individuals dynamically switch between ITN use and non-use strategies according to payoff-driven imitation dynamics. The payoff depends on prevalence and vector density, making it dynamic. Vector biting patterns vary with changes in host accessibility induced by protection coverage. The resulting feedback between host behavioral adaptation, vector biting redistribution, and disease prevalence generates rich nonlinear dynamics, including oscillatory outbreaks associated with Hopf bifurcation and critical transitions between disease states arising through saddle-node bifurcation. Our work provides a framework to understand how adaptive protection behavior and vector biting redistribution jointly shape disease dynamics in multi-host systems, and give rise to dynamical outcomes that cannot be captured by static transmission models.
	
The structure of this paper is as follows. In Section \ref{Mathmodel}, we formulate a multi-host vector-borne disease model incorporating adaptive ITN-use behavior and coverage-dependent vector biting preference. Section \ref{modanaly} provides a detailed analysis of the model, where conditions for the existence and stability of equilibria are derived. We also establish the conditions for the occurrence of saddle-node and Hopf bifurcations in this section. In Section \ref{numsec}, numerical simulations are presented to illustrate the dynamical behavior of the system and to probe the effect of disease risk, cost of ITN use, and vector biting behavior on long-term disease dynamics. Section \ref{conclu} discusses the conclusions and elaborates on the biological implications of the model.

	\section{Mathematical model}\label{Mathmodel} 
	We consider a system comprising two host species, $h1$ and $h2$, and a vector population, $v$, and develop a coevolutionary game-theoretic framework to investigate the control of a VBD. Let $N_{h1}$ denote the total population size of host $h1$, which is assumed to be constant. At any time $t$, $N_{h1}$ is divided into three classes: Susceptible $S_{h1}(t)$, Infected $I_{h1}(t)$, and Recovered $R_{h1}(t)$.  Infected individuals of host $h_1$ recover at a per capita rate $\mu_1$, and recovered individuals of host $h_1$ lose immunity and return to the susceptible class at a per capita rate $\delta_1$. Similarly, the total population of host $h2$ is also assumed to be a constant $N_{h2}$ and is divided into susceptible and infected compartments. At any time $t$, the total number of susceptible and infected are denoted as $S_{h2}(t)$ and $I_{h2}(t)$, respectively. The infected host $h2$ individuals recover and become susceptible again at a per capita rate of $\mu_2$. Each host species $hi,i=1,2$ has equal per capita birth and death rates $d_i$, which maintains the constant host population sizes. The total vector population at any time $t$ is represented by $M(t)$, and follows a logistic growth rate. At any time $t$, the vector population is classified into two compartments: susceptible $S_v(t)$ and infected $I_v(t)$. The parameters $b_3$ and $b_{31}$ denote the natural birth rate and the crowding effect on birth rate of the vector population, respectively, while  $d_3$ denotes the death rate of the vector. Throughout the study, assume the growth rate of the vector population is $g=b_3-d_3>0$. Let the number of bites per vector on a host per unit of time be $c$. Assume that the vector exhibits a feeding preference for $h1$ over the other $h2$, which is denoted by $\alpha_v > 0$. If the vector prefers to feed on $h1$, then $\alpha_v>1$, if the vector has equal feeding preference, then $\alpha_v=1$, otherwise, if the vector has a higher preference for $h2$, then $\alpha_v<1$. Let the parameter $\beta_{vh}$ indicate the transmission probability from the infected host to susceptible vectors, and let $\beta_{hv}$ represent the disease transmission probability from the infected vector to the susceptible host, respectively.

	\subsection{ VBD Dynamics with ITN Protection}\label{sec_stone}
	To integrate ITN use in the VBD model, we assume that only $h1$ individuals can choose whether to protect themselves against VBD with ITN. Let $\sigma_{hi}$ be the probability that vector attacks host $hi$ if encountered, and let the encounter rate of $hi$ be $\mathcal{E}_{hi}$, $i=1,2$. If $\theta$ represents the proportion of the $h1$ population using ITN, then  from \cite{stone2018evolution}, we find the bite rate of the vector as
	\begin{align} \label{eqc}
		c(\theta)=	\frac{1-\rho(\theta)}{\tau_f + \tau_2 (1-\rho(\theta))}.
	\end{align}
	Here, $\tau_f$ is the duration of a typical foraging bout, $\tau_2$ is the duration of the resting stage per day, and $\rho(\theta)$ represents the probability that the vector repeats or returns for another foraging bout. The time required to complete a feeding cycle will increase as ITN coverage increases; hence, the overall biting rate will reduce \cite{le2007elaborated}. The daily vector birth rate is given as $b_3(\theta)=f_v c(\theta)$, where $f_v$ is the vector fecundity. If the probability of successfully feeding in the presence of ITN is
	$W(\theta)$, then the probability of a vector successfully obtaining feed can be found as $W(\theta)/(1-\rho(\theta))$. Hence, if we consider $s_r$ as the probability of surviving the resting stage, then the vector death rate can be considered as
	\begin{align}\label{eqnford3}
		d_{3}(\theta)= \frac{c(\theta)}{2}\left(\frac{1-\rho(\theta)}{W(\theta) s_r}-1\right),
	\end{align} 
	 The feeding preference of the vector can be obtained by considering the ratio of the observed number of meals on $h1$ and $h2$ divided by the total number of $h1$ and $h2$ individuals. Hence, in the presence of ITN, the vector preference is given by 
	\begin{align} \label{eqnforav}
		\alpha_v(\theta)=\frac{(1-\varphi) (1-\theta+ s\theta )N_{h2} }{ N_{h1}  W(\theta) \varphi}.
	\end{align}
    Note that in the absence of ITN coverage among $h1$, the vector feeding preference will not change and will be equal to the innate vector preference. Hence, the innate vector feeding preference is given as $\alpha_v(0)$. For more details regarding vector mortality and vector preference in the presence of ITN, see Appendix \ref{detail_appen}. Throughout our study, we will assume that vectors will always feed on $h1$ if encountered, i.e., $\sigma_{h1}=1$.

	\subsection{Evolutionary game dynamics of ITN use}
To capture the adaptive evolution of ITN coverage among $h1$ individuals, we use an evolutionary game-theoretic framework in which individuals switch strategies through payoff-based imitation dynamics. The players in this game are $h1$ individuals, who can choose between two strategies: using or not using ITN protection. ITN use lowers the infection risk by reducing vector bites, but incurs the cost associated with its use. When the benefits of ITN use outweigh its cost, more individuals adopt the strategy of using ITN. To study how the number of ITN users changes in the population, we determine the expected payoff by considering the perceived cost of ITN and comparing it against the perceived disease risk and total vector population. Let $m$ represent the perceived cost of using ITN for a $h1$, then the payoff for an individual using ITN is 
	$$  \Pi_u =-m.$$
	Let $C_1$ and $C_2$ be the proportionality constants relating to the perceived cost associated with the risk of infection and vector population. Then the payoff for an individual not using ITN is
	$$  \Pi_{n}= -\big(C_1 I_{h1} + C_2 M\big).$$  
	The gain in payoff for an $h1$ individual using ITN is given by,
	$$ \Delta \Pi = \Pi_{u}-\Pi_{n}.  $$	  
	 When $\Delta \Pi>0$, the rate at which an $h1$ will choose the strategy to use ITN is proportional to $p\; \Delta \Pi$, where $p$ represents the imitation rate. Following the approach in \cite{hofbauer1998evolutionary}, the evolutionary dynamics of the proportion of $h1$ with ITN coverage is given by the following equation.
	\begin{equation}\label{repleqn}
		\frac{d\theta}{\text{d}t}=p \theta (1-\theta) \big(C_1 I_{h1} + C_2 M -m \big).
	\end{equation}
	The above equation is known as the {\em replicator equation}, which also governs the dynamics of the proportion of individuals not using ITNs $(1-\theta)$. 
	
	\subsection{Coupled Disease-ITN Co-evolutionary Dynamics} \label{modanaly}
	To investigate the effects of evolving ITN use on the prevalence of the disease, we integrate equation \eqref{repleqn} into the two-host VBD model similar to \cite{shravani&sharma}. 
    The dynamics of VBD after incorporating the ITN use intervention are governed by the following nonlinear differential equations:
	
	\begin{align}\label{eqn1}
		\begin{split}
			\diff{S_{h1}}{t}&= d_{1} N_{h1}-\frac{\alpha_{v}(\theta) c(\theta) \beta_{hv} I_{v}  S_{h1}}{\alpha_{v}(\theta) N_{h1}+N_{h2}} + \delta_{1} R_{h1}-d_{1} S_{h1},\\
			\diff{I_{h1}}{t}&=\frac{\alpha_{v}(\theta) c(\theta) \beta_{hv} I_{v} S_{h1}}{\alpha_{v}(\theta) N_{h1}+N_{h2}}-(\mu_1 + d_{1}) I_{h1},\\
			\diff{R_{h1}}{t}&=\mu_1  I_{h1}-(\delta_{1}+d_{1})  R_{h1},\\
			\diff{S_{h2}}{t}&=d_{2} N_{h2}-\frac{ c(\theta) \beta_{hv}  I_{v} S_{h2}}{\alpha_{v}(\theta) N_{h1}+N_{h2}}+ \mu_2 I_{h2}-d_{2} S_{h2} ,\\
			\diff{I_{h2}}{t}&=\frac{  c(\theta) \beta_{hv}  I_{v}  S_{h2}}{\alpha_{v}(\theta) N_{h1}+N_{h2}}-(\mu_2+d_{2}) I_{h2},\\
			\diff{S_{v}}{t}&=(b_3(\theta)-b_{31} (S_v + I_v)) (S_v + I_v)-\frac{c(\theta) \beta_{vh}  (\alpha_{v}(\theta)  I_{h1}+  I_{h2}) S_{v}}{\alpha_{v}(\theta) N_{h1}+N_{h2}}-d_3(\theta) S_{v},\\
			\diff{I_{v}}{t}&=\frac{c(\theta) \beta_{vh} (\alpha_{v}(\theta) I_{h1}+ I_{h2}) S_{v}}{\alpha_{v}(\theta) N_{h1}+N_{h2}}-d_3(\theta) I_{v},\\
			\diff{\theta}{t}&=p \theta (1-\theta) \big( C_1 I_{h1} + C_2 (S_v + I_v) -m \big).
		\end{split}
	\end{align}
	$S_{h1}(0)>0,I_{h1}(0)>0,R_{h1}(0)\geq 0, S_{h2}(0)>0,I_{h2}(0)\ge0,S_{v}(0)>0,I_{v}(0)\ge0, 1>\theta(0)>0$. 
	
Here, the parameters $\beta_{vh}$ and $\beta_{hv}$ represent the transmission probabilities from infected hosts to susceptible vectors and from infected vectors to susceptible hosts, respectively. All parameters are assumed to be finite and positive.

\section{Model analysis}	
	
	\subsection{Positivity and Boundedness of Solutions}	   
		For population models system such as \eqref{eqn2}, the non-negativity and boundedness of the solutions guarantees that all state variables remain meaningful as population sizes and populations remain finite for all future times. These properties imply the existence of a compact positively invariant region, which is the feasible region of the model system \eqref{eqn2}.

	\begin{proposition}
		Consider the set 
		\begin{align*} 
			\Gamma=\Big\{(S_{h1},I_{h1},R_{h1},S_{h2},I_{h2},S_v,I_{v},\theta)\in \mathbb{R}^{8} : 0\le S_{h1}(t)+I_{h1}(t)+R_{h1}(t) = N_{h1},0 \le S_{h2}+ I_{h2} =  N_{h2},\\ 0\le S_v +I_v \le  \frac{( b_{30} - d_{30})}{b_{31}}=K(0), 0\le \theta \le 1  \Big\}.\end{align*}
			\noindent Then $\Gamma$ is positively invariant for all solutions of the model system \eqref{eqn1} that start with non-negative initial conditions. 
	\end{proposition}
	\begin{proof}
			For any $S_{h1}\ge 0, I_{h1} \ge 0$, $R_{h1} \ge 0$, $S_{h2} \ge 0, I_{h2} \ge 0$, $S_{v} \ge 0$, $I_{v} \ge 0$, $\theta \ge 0$, the following holds,
		\begin{align*}
			\diff{S_{h1}}{t} \bigg\rvert_{S_{h1}=0}&= d_{1} N_{h1} + \delta_{1} R_{h1}\ge 0,\quad
			\diff{I_{h1}}{t} \bigg\rvert_{I_{h1}=0}=\frac{\alpha_{v} c \beta_{hv} I_{v} (N_{h1}-R_{h1})}{\alpha_{v} N_{h1}+N_{h2}} \ge 0,\quad 
			\diff{R_{h1}}{t}\bigg\rvert_{R_{h1}=0}=\mu_1  I_{h1}  \ge 0,\\
			\diff{S_{h2}}{t}  \bigg\rvert_{S_{h2}=0}&=d_{2} N_{h2}+ \mu_2 I_{h2} \ge 0,\quad
			\diff{I_{h2}}{t}\bigg\rvert_{I_{h2}=0}=\frac{c \beta_{hv}  I_{v}  N_{h2}}{\alpha_{v} N_{h1}+N_{h2}} \ge 0,\\
			\diff{S_{v}}{t} \bigg\rvert_{S_{v}=0}&=(b_3-b_{31}  I_v)  I_v \ge 0,\quad
			\diff{I_{v}}{t}\bigg\rvert_{I_{v}=0}=\frac{c \beta_{vh} (\alpha_{v} I_{h1}+ I_{h2}) S_v}{\alpha_{v} N_{h1}+N_{h2}}  \ge 0, \quad \diff{\theta}{t}\bigg\rvert_{\theta=0}=0.
		\end{align*}
		From above, we see that all the above solutions starting at the boundary of the non-negative region of $\mathbb{R}^{8}$,
		remain non-negative. Furthermore, due to the vector field pointing inside at every boundry plane, the uniqueness of the solutions guarantees that if the initial conditions are non-negative,  then the solution trajectory of each $S_{h1}, I_{h1}, R_{h1},S_{h2},I_{h2},S_v,I_{v}$ and $\theta$  will	persistently remain inside the non negetive region of $\mathbb{R}^{8}$. Therefore, the nonnegative cone of $\mathbb{R}^{8}$ remains positively invariant for the model system \eqref{eqn1}.
		Consider $N_{h1}=S_{h1}+I_{h1}+R_{h1}$ and  $N_{h2}=S_{h2}+I_{h2}$, then we have the following
		\begin{align*}
		    \diff{N_{h1}}{t}&=(d_1 - d_1)N_{h1}=0 \implies N_{h1}(t)=N_{h1}(0), t \ge 0,\\
             \diff{N_{h2}}{t}&=(d_2 - d_2) N_{h2}=0 \implies N_{h2}(t)=N_{h2}(0),  t \ge 0,
		\end{align*}
		Hence we obtain $S_{h1}(t)+I_{h1}(t)+R_{h1}(t)=N_{h1}(0)$ and $S_{h2}(t)+I_{h2}(t)=N_{h2}(0)$ for all $t \ge 0$. Now consider $M=S_v+I_v$, then we get 
		\[\diff{M}{t}=\big(g(\theta)-b_{31}M \big)M.\]
We have the following from model system \eqref{eqn1}
		\begin{align*}
			\diff{M}{t}& \le g(\theta) M - b_{31} M^2,\\
			\diff{\theta}{t}& \le p (C_1 N_{h1} + C_2 \; M ) (\theta -\theta^2).  
		\end{align*}
		Using the comparison theory we get
		\begin{align*}
			\lim_{t \rightarrow \infty} \sup  M(t) \le & \frac{( b_{30} - d_{30})}{b_{31}}=K(0),\\
			\lim_{t \rightarrow \infty} \sup  \theta(t) \le & 1.
		\end{align*}
		Here $b_{30} \coloneqq b_3(\theta)\big\rvert_{\theta=0}$, $d_{30} \coloneqq d_3(\theta)\big\rvert_{\theta=0}$.
		This is because $\frac{\partial c}{\partial \theta}<0$ and hence $c (\theta)\le c(0)$, therefore we have $b_3(\theta) \le b_3(0)$ for all $\theta \in [ 0,1]$.  We also assume vector mortality rate to increase with ITN coverage, that is $\frac{\partial d_3}{\partial \theta}>0$, hence $d_3(\theta) \ge d_3(0)$. This shows the boundedness for solutions of model \eqref{eqn1}.
	\end{proof}

\subsection{Disease-Free Equilibria and Reproduction Numbers} \label{Append:eqlbm}	
	The model \eqref{eqn1} is well posed since the functions on the right-hand side and their derivatives are continuous, ensuring the existence and uniqueness of solutions for given initial conditions. The model system \eqref{eqn1} satisfies the following
    \[ \diff{N_{h1}}{t}=(d_1 - d_1)N_{h1}=0,\quad  \diff{N_{h2}}{t}=(d_2 - d_2) N_{h2}=0. \]
Hence, $N_{h1}$ and $N_{h2}$ are constants for all $t\ge0$. It follows that $S_{h1}(t)=N_{h1}-I_{h1}(t)-R_{h1}(t)$, $S_{h2}(t)=N_{h2}-I_{h2}(t)$ and $S_{v}(t)=M(t)-I_{v}(t).$ Therefore, model system \eqref{eqn1} can be rewritten in the reduced form:
\begin{align}\label{eqn2}
		\begin{split}
			\diff{I_{h1}}{t}&=\frac{\alpha_{v}(\theta) c(\theta) \beta_{hv} I_{v} (N_{h1}-I_{h1}-R_{h1})}{\alpha_{v}(\theta) N_{h1}+N_{h2}}-(\mu_1 + d_{1}) I_{h1},\\
			\diff{R_{h1}}{t}&=\mu_1  I_{h1}-(\delta_{1}+d_{1})  R_{h1},\\
			\diff{I_{h2}}{t}&=\frac{  c(\theta) \beta_{hv}  I_{v}  (N_{h2}-I_{h2})}{\alpha_{v}(\theta) N_{h1}+N_{h2}}-(\mu_2+d_{2}) I_{h2},\\
			\diff{I_{v}}{t}&=\frac{c(\theta) \beta_{vh} (\alpha_{v}(\theta) I_{h1}+ I_{h2}) (M-I_{v})}{\alpha_{v}(\theta) N_{h1}+N_{h2}}-d_3(\theta) I_{v},\\
            \diff{M}{t}&=\big(g(\theta)-b_{31}M \big) M,\\
			\diff{\theta}{t}&=p \theta (1-\theta) \big( C_1 I_{h1} + C_2 M -m \big).
		\end{split}
	\end{align}
$I_{h1}(0)>0,R_{h1}(0)\geq 0,I_{h2}(0)\ge0,I_{v}(0)\ge0,M(0)>0, 1>\theta(0)>0$.
\begin{remark}
 In the absence of ITN use (\(\theta=0\)), the vector population follows logistic growth with intrinsic carrying capacity
\(
K(0)=\frac{g(0)}{b_{31}}>0.
\)
For \(\theta>0\), the net growth rate depends on the level of ITN coverage, resulting in a carrying capacity
\(
K(\theta)=\frac{g(\theta)}{b_{31}}>0.
\)
Therefore, ITN coverage modifies the carrying capacity of the vector population through its effect on the net growth rate \(g(\theta)\).
\end{remark}        
	
There exist two boundary, i.e, $E_{01}$, $E_{02}$, and three interior disease-free equilibria (DFE) i.e, $E_{03}$, $E_{04}$ and $E_{05}$, for the model system \eqref{eqn2}. 
	 The equilibria $E_{01}=(0,0,0,0,0,0)$ and $E_{02}=(0,0,0,0,0,1)$ correspond to the absence of vector and VBD, both with zero and complete ITN use, respectively. Note that $E_{02}$ is biologically implausible, since full ITN coverage is unlikely to be maintained in the absence of vectors.
     The equilibria $E_{03}=(0,0,0,0,K(0),0)$ and $E_{04}=(0,0,0,0,K(1),1)$ correspond to the presence of vector in disease-free states with zero and complete ITN use, respectively.
     The  equilibrium  $E_{0 5}=(0,0,0,0,K(\theta^*),\theta^*),\; \theta^* \in (0,1)$ denotes partial ITN usage in an entirely susceptible population with the presence of vector population. For the existence of $E_{0 5}$, we must get a $\theta^*\in (0,1)$ such that
   \begin{align}\label{DFE_poly}
    &A_1 \theta^{*2} + A_2 \theta^* +A_3=0,
    \end{align}    
   where, 
     \begin{align*}
        A_1=& \rho_2\left(l_2 s_r\left(\frac{2 b_{31}m \tau_2}{c_2} -(1+2f_v)\right)+\rho_2\right),\\
        A_2= &\Big((1+2f_v)s_r l_2+\rho_2\Big)(1-\rho_1)- \Big((1+2f_v)s_r l_1-1 +\rho_1\Big)\rho_2 -  \frac{2 s_r b_{31}m}{c_2} \Big(l_2(\tau_f+\tau_2(1-\rho_1))- l_1\tau_2\rho_2\Big),\\
        A_3= &\Bigg(\big(1-\rho_1\big)\Big(2f_v+ 1-\frac{2 b_{31}m}{c_2}\Big)-\frac{2 b_{31}m\tau_f}{c_2}\Bigg)l_1 s_r - \big(1-\rho_1\big)^2,\\
        \text{with~} & l_1= F p\Big(\varphi p_{h2}+(1-\varphi)p_{h1}\Big), \, l_2= F p(1-\varphi) p_{h1} \big(s(1-r)-1\big), \,\rho_1= (1-F)p,\quad  \rho_2= F p(1-\varphi)r
    \end{align*}
    
Hence, the equilibrium point $E_{05}$ exists whenever $A_3 (A_1+A_2+A_3)<0$, since under this condition equation \eqref{DFE_poly} has a positive root between 0 and 1.  To investigate the local stability of the disease-free equilibrium (DFE), we evaluate the Jacobian matrix of system \eqref{eqn2} at $(0,0,0,0,M^*,\theta^*),\theta^*\in[0,1]$, which is given as follows
    \begin{align*}
		J_{\text{DFE}}&=	\begin{bmatrix}
		- \mu_1 - d_1 &           0&0& \frac{N_{h1}  \beta_{hv} \alpha_{v}^* c ^*}{(\alpha_{v}^* N_{h1} + N_{h2})}& 0 & 0\\
		\mu_1 & - \delta_1 - d_1 &  0& 0& 0& 0\\
		0&  0&   - \mu_2- d_2 & \frac{N_{h2}  \beta_{hv} c^*}{(\alpha_{v}^* N_{h1} + N_{h2})}& 0& 0\\
		\frac{M^* \beta_{vh} \alpha_{v}^* c^*}{(\alpha_{v}^* N_{h1} + N_{h2})}&  0&  \frac{M^* \beta_{vh}  c^*}{(\alpha_{v}^* N_{h1} + N_{h2})}&  - d_{3}^*& 0 & 0\\
        0&   0&  0& 0&  b_3^*-d_3^* -2 b_{31} M^*& (f_v c^{\prime}-d_3^{\prime})M^*\\
		C_1 p \theta^* (1-\theta^*)&   0&  0& 0&p\theta^*(1-\theta^*)C_2 &   p(1-2\theta^*)(C_2  M^*-m)\\
	\end{bmatrix}.
\end{align*}
Here, \begin{align*}
c^{\prime}&=\frac{-\tau_{f}Fp_f(1-\varphi)r}{(\tau_f + \tau_2 (1-\rho^*))^2},\;
	d_3^{\prime}=\frac{c'}{2}\left(\frac{1-\rho^*}{W^* s_r}-1\right)+ \frac{c^*}{2 s_r W^{*2}} (-W^* \rho'-(1-\rho^*) W'),\\
    	  W'&=F p_f (1-\varphi)p_{h1}(s(1-r)-1), \;
	  \rho'=F p_f (1-\varphi)r.
      \end{align*}
	One eigenvalue for $J_{DFE}$ is $\lambda_1=-(\delta_1+d_1)<0$. 
    The other two eigenvalues of $J_{\text{DFE}}$ evaluated at $E_{01} $ are $\lambda_2=g(0)>0$ and $\lambda_3=-p m$ hence it is unstable.     
   Moreover, the two other eigenvalues of $J_{\text{DFE}}$ evaluated at $E_{02}$ are $\lambda_2=g(1)>0$ and $\lambda_3=p m>0$ therefore it is also unstable.
 To analyse the stability of the other three DFE, we follow the next generation matrix approach \cite{diekmann2010construction}. A detailed derivation of the  dominant eigenvalue of $J_{\text{DFE}}$ evaluated at $(0,0,0,0,K(\theta^*),\theta^*)$ is outlined in Appendix \ref{R0NGM}.  We find that DFE $(0,0,0,0,K(\theta^*),\theta^*)$ will be locally asymptotically stable if 
$$R_{c}(\theta^*)= \frac{\beta_{hv}\beta_{vh}c^{2}(\theta^*) K(\theta^*)}{(\alpha_{v}(\theta^*) N_{h1}+N_{h2})^2 d_3(\theta^*)}\left(\frac{\alpha_{v}^{2}(\theta^*) N_{h1}}{(\mu_{1} + d_1 )} +\frac{N_{h2}}{(\mu_2+d_{2})}\right)<1.$$ 
The threshold $R_c$ represents the {\em control reproduction number}, and is defined as the expected number of secondary vector-borne infections produced by a single infected host in an entirely susceptible population, when ITN coverage among $h1$ is $\theta^*$. 
Hence, the equilibrium $E_{04}$ and $E_{05}$, that is DFE corresponding to $(0,0,0,0,K(\theta^*),\theta^*), \; \theta^*\in(0,1]$,  is locally asymptotically stable if $R_c(\theta^*)<1.$
 Similarly, by evaluating $J_{\text{DFE}}$ at $E_{03}$, we find that $E_{03}$ will be locally asymptotically stable if 
    $$R_{0}= \frac{\beta_{hv}\beta_{vh}c^{2} K(0)}{(\alpha_{v}(0) N_{h1}+N_{h2})^2 d_3(0) }\left(\frac{\alpha_{v}^{2}(0) N_{h1}}{(\mu_{1} + d_1 )} +\frac{N_{h2}}{(\mu_2+d_{2})}\right)<1.$$ 
     Here, $R_0$ is the basic reproduction number in the absence of ITN use, which is similar to the expression obtained in \cite{shravani&sharma}.  The existence and stability of DFE are summarised in the following theorem.
\begin{theorem}
The model system \eqref{eqn2} admits the following disease-free equilibria:
\begin{itemize}
    \item[(i)] $E_{01}=(0,0,0,0,0,0)$, which always exists and is unstable.

    \item[(ii)] $E_{02}=(0,0,0,0,0,1)$, which always exists and is unstable.

    \item[(iii)] $E_{03}=(0,0,0,0,K(0),0)$, which always exists and is locally asymptotically stable if $R_{0}<1$.

    \item[(iv)] $E_{04}=(0,0,0,0,K(1),1)$, which always exists and is locally asymptotically stable if $R_{c}(1)<1$.

    \item[(v)] $E_{05}=(0,0,0,0,K(\theta^*),\theta^*)$, where $\theta^*\in(0,1)$, exists if
    $A_3(A_1+A_2+A_3)<0$ and is locally asymptotically stable if
    $R_c(\theta^*)<1$.
\end{itemize}  
\end{theorem}    
Parameter sets resulting in stability of various DFE are shown in the numerical example in Figure \ref{fig:1}.
	\subsection{Endemic equilibrium}
The model system \eqref{eqn2} has following endemic equilibria (EE), $E^*_1=(I_{h1}^{*},R_{h1}^{*},I_{h2}^{*}, I_{v}^{*},K,0)$ and $E^*_2=(I_{h1}^{*},R_{h1}^{*},I_{h2}^{*}, I_{v}^{*},K(1),1)$ which denote zero and complete ITN coverage in the presence of VBD, respectively. The model system also exhibits interior equilibrium relating to partial ITN adherence, $E^*=(I_{h1}^{*},R_{h1}^{*},I_{h2}^{*}, I_{v}^{*},K(\theta^*),\theta^{*})$ where $\theta^{*}\in (0,1)$. The components of these EE are given by,
	\begin{align*}
		I_{h1}^*&=\frac{w_1^* \alpha_{v}^* N_{h1} I_v^*}{w_1^* \alpha_{v}^* w_3 I_v^* + \mu_1 +d_1},\quad R_{h1}^*=\frac{\mu_1 I_{h1}^* }{\delta_1 +d_1} , \quad
		I_{h2}^*=\frac{w_1^* N_{h2} I_v^*}{w_1^* I_v^*  +\mu_2 +d_2},  \quad M^*=\frac{g(\theta^*)}{b_{31}}=K(\theta^*).
	\end{align*}
    Here $I_{v}^*$ satisfies the following
 \begin{align}\label{poleqnIv}  B_1 I_v^{*2} +B_2 I_v^* + B_3=0. \end{align}  
 Where
        	\begin{align*}
		B_1 & = w_1^{*2} w_2^* \alpha_v^{*2} N_{h1} + w_1^{*2} \alpha_v^* w_3 (w_2^*  N_{h2}+ d_3^*), \\
		B_2 &= w_1^*( d_3^*  (\mu_1+d_1) - w_1^* w_2^* M^*  \alpha_v^{*2} N_{h1} ) +  w_1^* \alpha_v^* w_3  (d_3^*  (d_2+\mu_2) \\& \quad - w_1^* w_2^* M^*  N_{h2})
		+ w_1^* w_2^* (\alpha_v^{*2} (d_2+\mu_2) N_{h1} + (\mu_1+d_1) N_{h2}),\\ 	
        B_3&=d_3^* (\mu_1+d_1)(d_2+\mu_2) \big(1 - R_c(\theta^*) \big),\\
		w_1(\theta)=&\frac{c(\theta) \beta_{hv}}{\alpha_v(\theta) N_{h1} +N_{h2}},\quad 
		w_2(\theta)=\frac{c(\theta) \beta_{vh}}{\alpha_v(\theta) N_{h1} +N_{h2}},\quad
		w_3=1+\frac{\mu_1 }{\delta_1 +d_1}.
	\end{align*}
	Observe that $M^*>0$ always since $K(\theta^*)>0$ for $\theta^*\in [0,1]$. Since $B_1>0$ hence by Descartes rule of signs if $R_c(\theta^*)>1$ then \eqref{poleqnIv} will have a unique positive real root $I_v^*>0$. 
	
	By evaluating the above at $\theta^*=0$, we find that $E_1^*$ will exist if $R_0>1$.  Similarly by evaluating the above at $\theta^*=1$, we find that $E_2^*$ will exist if $R_c(1)>1$. Also, if  $R_c(\theta^*)>1$ for $\theta^*\in (0,1)$ then one can see that $E^*$ will exist.
    
	Note that $\alpha_v^*,c^*,d_3^*,w_1^*,w_2^*$ are values of $\alpha_v(\theta),c(\theta),d_3(\theta), b_3(\theta),w_1(\theta),w_2(\theta)$ evaluated at $\theta^*$ respectively as described in Section \ref{sec_stone}. Here $\theta^* \in (0,1)$ is the root of the equation $G(\theta)=0$, where
\begin{equation} \label{eqntheta:01}
G(\theta)=\mathcal{G}_1 (\theta)  \left( \frac{ m-C_2 K(\theta)}{C_1}\right)^{2} + \mathcal{G}_2(\theta)  \left( \frac{m-C_2 K(\theta)}{C_1}\right)+ \mathcal{G}_3(\theta),
\end{equation}
      \text{ where,}      
      \begin{align*}
		\mathcal{G}_1&=\alpha_{v}(\theta)  w_2(\theta) ((\mu_{1}+d_1)+ w_1(\theta) \alpha_{v}(\theta)  K(\theta)  w_3 ) ((\mu_2+d_2) \alpha_{v}(\theta) w_3 -(\mu_{1}+d_1)), \nonumber\\
		\mathcal{G}_2&= \big( \mu_{1}+d_1-(\mu_2+d_2) \alpha_v(\theta) w_3\big) \big(w_1(\theta) w_2(\theta) \alpha_v^{2}(\theta) N_{h1} K(\theta)-d_3(\theta) (\mu_{1}+d_1)  \big) \\&-  w_2(\theta) (\alpha_v^{2}(\theta) N_{h1} (\mu_2+d_2) +N_{h2} (\mu_{1}+d_1)) \big(\mu_{1}+d_1 + w_1(\theta) \alpha_v(\theta)   K(\theta) w_3\big),\\
		\mathcal{G}_3&=\alpha_{v} N_{h1}  d_{3}(\theta)   (\mu_{1}+d_1) (\mu_2+d_2)  \big(R_c(\theta)-1\big). 
	\end{align*}

\noindent Note that if $G(0)G(1)<0,$ then by the intermediate value theorem there exists at least one root $\theta^* \in(0,1)$ such that $G(\theta^*)=0.$ This condition is numerically verified in Figure \ref{fig:7} of Appendix \ref{figureapp}, where $G(0)G(1)=-398.45<0$ and the corresponding root of $G(\theta)=0$ is $\theta^*=0.89\in(0,1).$ The existence result for the endemic equilibrium (EE) is summarized below.
    
     \begin{theorem}
     The model \eqref{eqn2} admits the following endemic equilibria
        \begin{itemize}
            \item[(i)] $E_{1}^*=(I_{h1}^{*},R_{h1}^{*},I_{h2}^{*}, I_{v}^{*},K(0),0)$ exists if $R_0>1$, 
            \item[(ii)] $E_{2}^*=(I_{h1}^{*},R_{h1}^{*},I_{h2}^{*}, I_{v}^{*},K(1),1)$ exists if $R_c(1)>1$, 
             \item[(iii)] $E^*=(I_{h1}^{*},R_{h1}^{*},I_{h2}^{*}, I_{v}^{*},K(\theta^*),\theta^*), \theta^* \in (0,1)$ exists if $R_c(\theta^*)>1$ and $G(0) G(1)<0$.
        \end{itemize}
     \end{theorem}  
\noindent Now to study the stability of EE we find the variational matrix of the model system \eqref{eqn2} evaluated at $(I_{h1}^{*},R_{h1}^{*},I_{h2}^{*}, I_{v}^{*},K(\theta^*),\theta^*)$ which is given as 
	\begin{align*}
		J_{EE}=	\begin{bmatrix}
			f_{11}&f_{12}&0&f_{14}&0&f_{16}\\
			\mu_{1}&f_{22}&0&0&0&0\\
			0&0&f_{33}&f_{34}&0&f_{36}\\
			f_{41}&0&f_{43}&f_{44}&f_{45}&f_{46}\\
            0&0&0&0&f_{55}&f_{56}\\
			f_{61}&0&0&0&f_{65}&f_{66}\\
		\end{bmatrix},	
	\end{align*}
	where, 
	\begin{align*}
	f_{11}&=f_{12} - \mu_{1} - d_1, 	
    f_{12}=-\frac{\beta_{hv} \alpha_{v}^* c^* I_{v}^*}{(\alpha_{v}^* N_{h1} + N_{h2})},
	f_{14}=\frac{\beta_{hv} \alpha_{v}^* c^*  S_{h1}^*}{(\alpha_{v}^* N_{h1} + N_{h2})},f_{22} = - \delta_1 - d_1, f_{61}=C_1 p \theta^* (1-\theta^*),\\ 
	f_{33}&= - \mu_2-d_2 - \frac{\beta_{hv} c^* I_{v}^*}{(\alpha_{v}^* N_{h1} + N_{h2})} ,
	f_{34}=\frac{\beta_{hv} c^* S_{h2}^*}{(\alpha_{v}^* N_{h1} + N_{h2})} ,	f_{41}=\frac{\beta_{vh} \alpha_{v}^* c^* S_v^*}{(\alpha_{v}^* N_{h1} + N_{h2})} ,f_{43}=\frac{\beta_{vh} c^* S_v^*}{(\alpha_{v}^* N_{h1} + N_{h2})},\\
	f_{44}&=- \frac{\beta_{vh} c^* (\alpha_{v}^* I_{h1}^* + I_{h2}^*) }{(\alpha_{v}^* N_{h1} + N_{h2})} - d_3^*,
	 f_{45}=\frac{\beta_{vh} c^* (\alpha^*_v I_{h1}^* +I_{h2}^*)}{(\alpha_{v}^* N_{h1} + N_{h2})},f_{16}=\frac{\beta_{hv} I_v^* S_{h1}^*}{(\alpha_{v}^* N_{h1} + N_{h2})^2} \big( \alpha_{v}^{*2} N_{h1} c^{\prime} + N_{h2} (c \alpha_{v})^{\prime}\big),\\
	f_{36}&=\frac{\beta_{hv} I_v^* S_{h2}^*}{(\alpha_{v}^* N_{h1} + N_{h2})^2} \Big(c^{\prime} (\alpha_{v}^* N_{h1} + N_{h2})-c^* \alpha_{v}^{\prime} N_{h1}  \Big),
   f_{56}=(f_v c^{\prime}-d_3^{\prime})M^*,
    f_{65}=C_2 p\theta^*(1-\theta^*),\\
	f_{46}&=\beta_{vh} S_v^* \left(  \frac{(\alpha_{v}^* I_{h1}^* + I_{h2}^*)  c^{\prime}}{(\alpha_{v}^* N_{h1} + N_{h2})} + \frac{c^* \alpha_{v}^{\prime} (I_{h1}^* N_{h2}-N_{h1} I_{h2}^*)}{(\alpha_{v}^* N_{h1} + N_{h2})^2} \right) - I_v^* d_3^{\prime}, f_{55}=-(b_3^*-d_3^*),
	\\&f_{66}= p (1 - 2 \theta^*)(C_1 I_{h1}^* + C_2 M^*-m ),\alpha_v'= \frac{(1-\varphi)N_{h2}}{\varphi N_{h1} W^{*2}}\big(W^*(s-1)-(1+\theta^*(s-1)W')\big).
	%
\end{align*}

At $\theta^*=0$ and $\theta^*=1$ we have $f_{61}=0=f_{65}$. In such a case, two of the eigenvalues of $J_{EE}$ will be $\lambda_1=f_{55}<0$ and $\lambda_2=f_{66}$. At $\theta^*=0$ we have $\lambda_2>0$ if $C_1 I_{h1}^*+C_2 M^*-m>0$ whereas at $\theta^*=1$ we get $\lambda_2>0$ if $C_1 I_{h1}^*+C_2 M^*-m<0$.
Hence we find that $E_1^*$ is unstable if $C_1 I_{h1}^*+C_2 M^*-m>0$ and $E_2^*$ is unstable if $C_1 I_{h1}^*+C_2 M^*-m<0$.\\
Now, if at $\theta^*=0$ we have $C_1 I_{h1}^*+C_2 K-m<0$, then $E_1^*$ is stable if all the roots of the following equation evaluated at $E_1^*$ have negative real part
    \[ \lambda_1^4 +N_1 \lambda^3 +N_2 \lambda^2 +N_3 \lambda +N_4=0,\]
 where,
    \begin{align*}
    \begin{split}
       N_1&=- f_{11} - f_{22} - f_{33} - f_{44},\\
        N_2&=-f_{12} \mu_1+f_{22} f_{33}-f_{14} f_{41}-f_{34} f_{43}+(f_{22}+f_{33}) f_{44}+f_{11} (f_{22}+f_{33}+f_{44}),\\ 
N_3&=-f_{22}f_{33}(f_{44}+f_{11})+(f_{11}+f_{22})f_{34}f_{43}+f_{12}\mu_1(f_{33}+f_{44})+(f_{22}+f_{33})(f_{14}f_{41}-f_{11}f_{44}),\\
N_4&=(f_{11}f_{22}-f_{12}\mu_1)(f_{33}f_{44}-f_{34}f_{43})-f_{14}f_{22}f_{33}f_{41}.
 \end{split}
    \end{align*}
	Since $N_1>0$ always, by the Routh Hurwitz criterion, we establish that all eigenvalues of $J_E^*$ evaluated at $E_1$ will have a negative real part if conditions \eqref{cond_EE1_stab2} are satisfied. 
    \begin{flalign}\label{cond_EE1_stab2}
\hspace{-5cm}	N_i>0, i=2,3,4, \quad N_1 N_2> N_3,\quad N_{3}(N_1 N_2-N_3)-N_4 N_1^2>0.
\end{flalign}
After simplification, one can see that conditions \eqref{cond_EE1_stab2} are satisfied if 
    \begin{align}\label{cond_EE_stab2}
    \begin{split}
    \frac{2 \beta_{hv} \beta_{vh} c^{*2} \alpha_v^{*2}  N_{h1} M^*}{(\alpha_v^{*} N_{h1}+N_{h2})^2}<d_3^* (\mu_1+d_1) \text{  and  }
    \frac{2 \beta_{hv} \beta_{vh} c^{*2}  N_{h2} M^*}{(\alpha_v^{*} N_{h1}+N_{h2})^2}<d_3^* (\mu_2+d_2). 
    \end{split}
\end{align}
Similarly, at $\theta^*=1$ if we get $C_1 I_{h1}^*+C_2 K(1)-m>0$, then $E^*_2$ is stable if the conditions \eqref{cond_EE_stab2} evaluated at $E^*_2$ are satisfied. 

  The characteristic equation associated with $J_{EE}$  of model \eqref{eqn2} evaluated at $E^*$ is given as
	\begin{align}\label{charceqn}
		\lambda ^6 + M_1 \lambda ^5 + M_2 \lambda ^4+ M_3 \lambda ^3+ M_4 \lambda^2 + M_5 \lambda +M_6=0.
	\end{align}
	The coefficients $M_i$'s for the above expression are stated in Appendix \ref{charceqn1}. Since $M_1>0$ always, by using the Routh Hurwitz criterion, we establish that all eigenvalues of $J_{EE}$ at $E^*$ will have negative real part if conditions \eqref{cond_EE_stab1} are satisfied. 
    \begin{align}\label{cond_EE_stab1}
	\begin{split}
		&M_j>0, j=2,\dots,6,\quad
		M_1 M_2>M_3,\quad
		M_3(M_1 M_2-M_3)>M_1 (M_1 M_4-M_5),\\
		&(M_{1} M_{2}-M_{3})(M_{1} M_{6} + M_{3} M_{4} -M_{2} M_{5})>(M_{1} M_{4} - M_{5})^2,\\
		& M_{5} \big(M_{1} M_{3} (M_{2} M_{4} - 3 M_{6}) - M_{1}^2 (M_{4}^2 - 2 M_{2} M_{6}) - M_{3}^2 M_{4} +M_{5} (M_{2} M_{3} + M_{1} (2  M_{4} - M_{2}^2))  - M_{5}^2\big) \\& + M_{6} M_{3} (M_{3}^2 + M_{1} (M_{1} M_{4} - M_{2} M_{3})) - M_{1}^3 M_{6}^2>0. 
	\end{split}
\end{align}
    We summarise the stability results of EE in the following theorem.
    
 \begin{theorem}
 The local stability of the endemic equilibria of system \eqref{eqn2} is characterized as follows:
\begin{itemize}
    \item[(i)] If $C_1 I_{h1}^* + C_2 K(0) - m < 0$, then the equilibrium $E_1^*$ is locally asymptotically stable provided that the conditions in \eqref{cond_EE_stab2}, evaluated at $E_1^*$, are satisfied.

    \item[(ii)] If $C_1 I_{h1}^* + C_2 K(1) - m > 0$, then the equilibrium $E_2^*$ is locally asymptotically stable provided that the conditions in \eqref{cond_EE_stab2}, evaluated at $E_2^*$, are satisfied.

    \item[(iii)] The interior endemic equilibrium $E^*$ is locally asymptotically stable if and only if the conditions in \eqref{cond_EE_stab1}, evaluated at $E^*$, are satisfied.
\end{itemize}
 \end{theorem}
    
	\subsection{Saddle node bifurcation}
	To establish that the model system \eqref{eqn2} undergoes saddle-node bifurcation, we select $\mu_1$ as the bifurcation parameter and consider the other parameters constant.  The, Jacobian matrix $J_{EE}$  has a simple eigenvalue 0, if $M_6(\mu_1^*) = 0$ where $\mu_1^*$ is the critical value at which bifurcation occurs. Let $E_{SN}^*$ represents the value of the equilibrium point $E^*$ at $\mu_1=\mu_1^*$.
    To show Saddle-node bifurcation occurs at $\mu_1=\mu_1^*$ we establish the transversality conditions \cite{perko}.
	Let $\eta_1=(\eta_{11},\eta_{12},\eta_{13},\eta_{14},\eta_{15},\eta_{16})^{T}$ and $\eta_2=(\eta_{21},\eta_{22},\eta_{23},\eta_{24},\eta_{25},\eta_{26}) $ be the right and left eigenvectors of the Jacobian matrix $J_{E_{SN}^*}\big\rvert_{\mu_{1}=\mu_1^*}$ corresponding to 0 eigenvalue, where
	\begin{align*}
		\eta_{11}&=\frac{1}{f_{61}}\bigg(\frac{f_{65} f_{56}}{f_{55}}-f_{66}\bigg),\quad \eta_{12}= \frac{-\mu_{1}}{f_{22} f_{61}}\bigg(\frac{f_{65} f_{56}}{f_{55}}-f_{66}\bigg), \\
		\eta_{13}&=\bigg(\frac{f_{41}}{f_{61}}\Big(\frac{f_{65}f_{56}}{f_{55}}-f_{66}\Big)-\frac{f_{45}f_{56}}{f_{55}}+ f_{46}-\frac{f_{44}f_{36}}{f_{34}}\bigg) \frac{f_{34}}{f_{44} f_{33}-f_{34} f_{43} },\\
        \eta_{14}&=\frac{1}{f_{14}}\bigg(\frac{1}{f_{61}}\Big(\frac{f_{12}\mu_1}{f_{22}}-f_{11}\Big)\Big(\frac{f_{65}f_{56}}{f_{55}}-f_{66}\Big)-f_{16}\bigg), \quad
         \eta_{15}=-\frac{f_{56}}{f_{55}}, \quad \eta_{16}=1,\\
        \eta_{21}&=\frac{1}{f_{14}} \Big( \frac{f_{43} f_{34}}{f_{33}} -f_{44}\Big),\quad
       \eta_{22}=-\frac{f_{12}}{f_{22}}\eta_{21}, \quad \eta_{23}=-\frac{f_{43}}{f_{33}},\quad \eta_{24}=1,\\
		\eta_{25}&= \left( \frac{f_{16}}{f_{14}}\Big( f_{44}- \frac{f_{43} f_{34}}{f_{33} f_{14}} \Big) + \frac{f_{66} f_{45}}{f_{65}}+\frac{f_{43} f_{36}}{f_{33}} -f_{46} + \right) \frac{f_{65} }
        { f_{56} f_{65}  - f_{55} f_{66}},\\
		\eta_{26}&=\frac{1}{f_{61}} \bigg( \frac{1}{f_{14}}\Big(f_{11} -\frac{\mu_1 f_{12}}{f_22}\Big) \Big(f_{44}-\frac{f_{43} f_{34}}{f_{33}}\Big) -f_{41} \bigg).
	\end{align*}
	The model system \eqref{eqn2} can be expressed as 
     \begin{align*} 
     \diff{X}{t}= Q= &[Q_1(X,\mu_1),Q_2(X,\mu_1),Q_3(X,\mu_1),Q_4(X,\mu_1),Q_5(X,\mu_1),Q_6(X,\mu_1)]^{T},\\
     &\text{ where } X=[I_{h1}, R_{h1}, I_{h2},I_v,M,\theta]^{T}. 
     \end{align*}
    We find the following transversality conditions: \\
$ \mathcal{Q}_1= \eta_2 \cdot \frac{\partial Q}{\partial \mu_{1}}\bigg\rvert_{E_{SN}^*(\mu_1=\mu_1^*)}=(\eta_{22}-\eta_{21})I_{h1}^*$,\\
$\mathcal{Q}_2= \eta_2 \cdot \left[D^2_{(I_{h1},R_{h1},I_{h2},I_{v},M,\theta)} Q(\eta_{1},\eta_{1})\right]\bigg\rvert_{E_{SN}^*(\mu_1=\mu_1^*)}=\eta_2 \cdot  \begin{bmatrix}
	k_1 &
	0&
	k_3&
	k_4&
	k_5&
	k_6
\end{bmatrix}^{T}$,\\
where, \begin{align*}
	k_1&=2 \eta_{14} \bigg( \eta_{11} \frac{\partial^2 Q_1}{\partial I_v \partial I_{h1}} +\eta_{12} \frac{\partial^2 Q_1}{\partial I_v \partial R_{h1}} \bigg)+
	2 \eta_{16} \bigg( \eta_{11} \frac{\partial^2 Q_1}{\partial I_{h1} \partial \theta}+ \eta_{12} \frac{\partial^2 Q_1}{\partial R_{h1} \partial \theta} + \eta_{14} \frac{\partial^2 Q_1}{\partial I_{v} \partial \theta} \bigg) + \eta_{16}^2 \frac{\partial^2 Q_1}{\partial  \theta^2 },\\
	k_3&=2 \eta_{14} \bigg( \eta_{13} \frac{\partial^2 Q_3}{\partial I_v \partial I_{h2}}  + \eta_{16} \frac{\partial^2 Q_3}{\partial I_v \partial \theta} \bigg) + 2 \eta_{13} \eta_{16} \frac{\partial^2 Q_3}{\partial I_{h2} \partial \theta} +\eta_{16}^2 \frac{\partial^2 Q_3}{\partial  \theta^2 },\\
	k_4&=2 \eta_{15} \bigg( \eta_{11} \frac{\partial^2 Q_4}{\partial M \partial I_{h1}} + 
	\eta_{13} \frac{\partial^2 Q_4}{\partial I_{h2} \partial M}  + \eta_{16} \frac{\partial^2 Q_4}{\partial M \partial \theta} \bigg) 
	+ 2 \eta_{14} \bigg( \eta_{11} \frac{\partial^2 Q_4}{\partial I_{h1} \partial I_{v}} + \eta_{13} \frac{\partial^2 Q_4}{\partial I_{h2} \partial I_{v}} + \eta_{16} \frac{\partial^2 Q_4}{\partial I_v \partial \theta}  \bigg)
	\\& + 2 \eta_{16} \bigg(\eta_{11} \frac{\partial^2 Q_4}{\partial I_{h1} \partial \theta} + \eta_{13} \frac{\partial^2 Q_4}{\partial I_{h2} \partial \theta}\bigg) + \eta_{16}^2 \frac{\partial^2 Q_4}{\partial  \theta^2 },\quad
	k_5=2 \eta_{15} \eta_{16} \frac{\partial^2 Q_5}{\partial M \partial \theta}+ \eta_{15}^2 \frac{\partial^2 Q_5}{\partial  M^2 }+ \eta_{16}^2 \frac{\partial^2 Q_5}{\partial  \theta^2 },\\
	k_6&= 2 \eta_{16} \bigg( \eta_{11} \frac{\partial^2 Q_6}{\partial I_{h1} \partial \theta} + \eta_{15} \frac{\partial^2 Q_6}{\partial M \partial \theta } \bigg) + \eta_{16}^2 \frac{\partial^2 Q_6}{\partial  \theta^2 }.
\end{align*}
If $\mathcal{Q}_1\ne0$ and $\mathcal{Q}_2\ne0$  then using Sotomayor's theorem \cite{perko}, we find that model system undergoes saddle-node bifurcation near interior equilibrium $E^*$. The result is summarised as follows.

\begin{theorem}
	The model system \eqref{eqn2} undergoes saddle-node bifurcation at $E^*$ if there exists $\mu_1=\mu_1^*$ such that $\mathcal{Q}_1(\mu_{1}^*)\ne 0$ and  $\mathcal{Q}_2(\mu_{1}^*)\ne 0$.
\end{theorem}

	The occurrence of saddle-node bifurcation is shown graphically in Figure \ref{fig:3}. The conditions are validated by numerically computing the value of $\mathcal{Q}_1=-39.74$ and $\mathcal{Q}_2=-4.92308665 \times 10^{-5}$ at $\mu_{1}^*=0.34347$.
    This shows that as the recovery rate of protected $h1$ decreases, the stable and unstable interior equilibrium points for the model system \eqref{eqn2} converge into a single stable point.
	
	\subsection{Hopf bifurcation}
	 Hopf bifurcation indicates the emergence and decay of periodic solutions due to small perturbations in the equilibrium point. We investigate the presence of the Hopf bifurcation in the model system \eqref{eqn2} by considering the perceived cost of ITN use, $m$, as the bifurcation parameter and considering other parameters constant.  
Let the critical value of $m$, denoted as $m^*$, be defined as the value at which $M_{5} \big(M_{1} M_{3} (M_{2} M_{4} - 3 M_{6}) - M_{1}^2 (M_{4}^2 - 2 M_{2} M_{6}) - M_{3}^2 M_{4} +M_{5} (M_{2} M_{3} + M_{1} (2  M_{4} - M_{2}^2))  - M_{5}^2\big)  + M_{6} M_{3} (M_{3}^2 + M_{1} (M_{1} M_{4} - M_{2} M_{3})) - M_{1}^3 M_{6}^2=0.$
Then, at $m=m^*$, the characteristic equation \eqref{charceqn} can be expressed as 
\begin{align*}
	&\left(\lambda^2 +\omega\right)\left(\lambda^4+M_1 \lambda^3+\left(M_2-\omega\right) \lambda^2 + (M_3 -\omega M_1)\lambda+ M_4 - M_2 \omega+ \omega^2 \right)=0,\\
	&\text{ with, } \omega=\big(u_2 + \sqrt{u_2^2-4u_1 M_1 M_6}\big)\frac{1}{2 u_1}, u_1=M_1 M_2-M_3,u_2=M_1M_4-M_5.
\end{align*}
Since $u_1>0$ from \eqref{cond_EE_stab1}, therefore, by Descartes Rule of signs, there exists a pair of purely imaginary roots $\lambda_{1,2}=\pm i \sqrt{\omega}$ if $u_2>0$. To show that a pair of complex conjugate eigenvalues cross the imaginary axis with non-zero rate, consider the value $m$ in the neighborhood of $m^*$. Let the roots of the characteristic equation with respect to $J_{EE}$ in this vicinity be denoted as $\lambda_{1,2}=\zeta(m) +i \xi(m)$. To establish the transversality condition, we substitute these into \eqref{charceqn} and find the derivative of the real and imaginary parts, respectively, given as
\begin{align*}
	& \quad \quad \quad \quad L_1 \dot{\zeta} -	L_2 \dot{\xi}+L_3=0,\quad \text{and} \quad
	L_2 \dot{\zeta} +	L_1 \dot{\xi}+L_4=0,\\	
	\text{where,}&\\
	L_1 &= 30 \zeta^4 \xi + 20 \zeta^3 M_1 \xi + 12 \zeta^2 M_2 \xi + 6 \zeta M_3 \xi + 2 M_4 \xi - 
	60 \zeta^2 \xi^3 - 20 \zeta M_1 \xi^3 - 4 M_2 \xi^3 + 6 \xi^5,\\
	L_2 &= 6 \zeta^5 + 5 \zeta^4 M_1 + 4 \zeta^3 M_2 + 3 \zeta^2 M_3 + 2 \zeta M_4 + M_5 - 
	60 \zeta^3 \xi^2 - 30 \zeta^2 M_1 \xi^2 - 12 \zeta M_2 \xi^2 \\&- 3 M_3 \xi^2 + 30 \zeta \xi^4 + 
	5 M_1 \xi^4,\\ 
	L_3 &= \zeta^5 \dot{M}_1 - 10 \zeta^3 \xi^2 \dot{M}_1 + 5 \zeta \xi^4 \dot{M}_1 + \zeta^4 \dot{M}_2 - 
	6 \zeta^2 \xi^2 \dot{M}_2 + \xi^4 \dot{M}_2 + \zeta^3 \dot{M}_3 - 3 \zeta \xi^2 \dot{M}_3 + 
	\zeta^2 \dot{M}_4 \\&- \xi^2 \dot{M}_4 + \zeta \dot{M}_5 +  \dot{M}_6,\\
	L_4 &= 5 \zeta^4 \xi \dot{M}_1 - 10 \zeta^2 \xi^3 \dot{M}_1 + \xi^5 \dot{M}_1 + 
	4 \zeta^3 \xi \dot{M}_2 -  4 \zeta \xi^3 \dot{M}_2 + 3 \zeta^2 \xi \dot{M}_3 - \xi^3 \dot{M}_3 + 
	2 \zeta \xi \dot{M}_4 + \xi \dot{M}_5.
\end{align*}
Note that here $\dot{x}=\diff{x}{m}$. Evaluating the above $L_i,i=1,\dots,4 $ at $m=m^*$ we get 
\begin{align*}
	\diff{\zeta}{m}\bigg\rvert_{m=m^*}=\frac{-(L_1 L_3+L_2 L_4)}{L_1^2 +L_2^2}\bigg\rvert_{m=m^*}.
\end{align*}
Hence, $\diff{\mathfrak{Re}(\lambda)}{m} \ne 0$ if $(L_1 L_3+L_2 L_4)\big\rvert_{m=m^*}\ne 0$, which results in following condition\\
$(2 M_4 \omega - 4 M_2 \omega^3 + 6 \omega^5) (\omega^5 \dot{M}_1 - \omega^3 \dot{M}_3 + 
\omega \dot{M}_5) + (M_5 - 3 M_3 \omega^2 + 5 M_1 \omega^4) (\omega^4 \dot{M}_2 - \omega^2 \dot{M}_4 + 
\dot{M}_6) \ne 0.$ The result is summarised as follows.

	\begin{theorem}
		At $m=m^*$, the system \eqref{eqn1} undergoes a Hopf bifurcation at $E^*$ if the following conditions are satisfied. 
		\begin{itemize}
			\item[i)]$M_{5} \big(M_{1} M_{3} (M_{2} M_{4} - 3 M_{6}) - M_{1}^2 (M_{4}^2 - 2 M_{2} M_{6}) - M_{3}^2 M_{4} +M_{5} (M_{2} M_{3} + M_{1} (2  M_{4} - M_{2}^2))  - M_{5}^2\big)  + M_{6} M_{3} (M_{3}^2 + M_{1} (M_{1} M_{4} - M_{2} M_{3})) - M_{1}^3 M_{6}^2=0.$
			\item[ii)] $(2 M_4 \omega - 4 M_2 \omega^3 + 6 \omega^5) (\omega^5 \dot{M}_1 - \omega^3 \dot{M}_3 + 
    \omega \dot{M}_5) + (M_5 - 3 M_3 \omega^2 + 5 M_1 \omega^4) (\omega^4 \dot{M}_2 - \omega^2 \dot{M}_4 + 
    \dot{M}_6) \ne 0,$\\
    $\omega=\big(u_2 + \sqrt{u_2^2-4u_1 M_1 M_6}\big)\frac{1}{2 u_1},\; u_1=M_1 M_2-M_3,\; u_2=M_1M_4-M_5>0$.
		\end{itemize}
	\end{theorem}
Note that as $m$ passes through $m^*$, the endemic equilibrium loses stability and a family of periodic solutions emerges. As a result of this the disease prevalence and the ITN use no longer converge to steady-state values but instead exhibit sustained oscillations. These oscillations arise from the feedback between disease transmission and adaptive protection behaviour, whereby changes in disease prevalence influence ITN usage, which in turn alters transmission dynamics. The existence of a Hopf bifurcation is verified numerically in Figure \ref{fig:4}. For $m^*=10.83$ and other the parameter values are specified in Figure \ref{fig:4}, we obtain the value of $\diff{\zeta}{m}\bigg\rvert_{m=m^*}=0.03426 \ne 0.$ 
	
	\section{Numerical results}\label{numsec}
	In this section, we present the numerical simulations of the co-evolutionary game theoretic VBD model \eqref{eqn2} to validate the analytical findings of this study. We also investigate the dynamical behavior of the model system for various parameter values. The parameter values used for simulating model solutions are consistent with the biological ranges provided in previous studies.

	\subsection{Stability regions for disease-free equilibria}
	To study the switching of stability of DFE corresponding to zero to complete ITN coverage in the presence of vector population, we numerically simulate the model to study the equilibrium solution $\theta^*$ when there is no disease in the population, but vector population is present. Figure \ref{fig:1}(a)-(c) demonstrates the parameter ranges resulting in stability of $\theta^*=0,\theta^*=1$ and $0<\theta^*<1$ in the absence of VBD, for various combinations of perceived cost of ITN, $m$, and encounter rates of $h1$, $\mathcal{E}_{h1}$.  It indicates that at the extreme low value of $m$, the population comprises only ITN-using individuals, and at a high value of $m$, only ITN non-using individuals. Whereas ITN users and non-users coexist with each other for intermediate values of $m$. This is because larger values of $m$ increase the perceived cost of using ITN against the cost of vector density, hence $h1$ individuals choose to stop using ITN. Further, as  $\mathcal{E}_{h1}$ increases, the threshold value of $m$ below which $E_{04}$ ($\theta^*=1$)  is stable decreases, but the threshold value of $m$ above which $E_{03}$ ($\theta^*=0$) is stable increases. We also see that as $\sigma_{h2}$ increases, the size of the parameter region resulting in stability of $E_{05}$ decreases, and the stability region for $E_{04}$ increases. To understand such behavior, we study how the density of vector population is affected with varying $m$ and $\mathcal{E}_{h1}$ values in Figure \ref{fig:1} (d)-(f).  Here, the region inside the green line corresponds to stability of $E_{05}$. Within this region, the value of $M^*$ increases with $m$ and is unaffected by varying $\mathcal{E}_{h1}$. However, as $\mathcal{E}_{h1}$ increases, $M^*$ deccreases when $E_{04}$ is stable, and $M^*$ increases when $E_{03}$  is stable. When $\theta^*=0$ the vector gets to have undisruputed feeds, hence their density is the highest in this region. Moreover, from \eqref{eqnforav} it follows that $\alpha_{v}$ decreases as $\sigma_{h2}$ increases, which implies that the probability of vector bites on unprotected $h2$ increases. The undeterred feeding on the $h2$ leads to an increase in the density of the vector. Hence, as $\sigma_{h2}$ increases, we observe that the cost induced by $M^*$ increases, incentivizing the use of ITN.

		\subsection{ Behavior of endemic equilibrium under changing $R_0$ and $m$ for different feeding preference}
        We present the behavior of EE of the model system \eqref{eqn2} as a function of perceived cost of ITN, $m$, and basic reproduction number, $R_0$, with low and high innate vector preference for $h2$ in the left and right panels of Figure \ref{fig:2}, respectively. Here, $R_0$ is varied by changing $h1$ recovery rate, $\mu_1$. For lower values of $m$ in Figure \ref{fig:2}, $E_2^*$ ($\theta^*=1$) is stable as the perceived cost of ITN is negligible against VBD. Moreover, the high values of $m$ result in $E_1^*$ ($\theta^*=0$) to be stable. As $m$ increases, the threshold value of $R_0$ below which $h1$ individuals stop using ITN increases. We see that model solution $I_{h1}^*$ is more sensitive to $m$ when $\theta^*\in (0,1)$, and the influence of $R_0$  is higher when $\theta^*= 0$ or $\theta^*= 1$.
        
            In Figure \ref{fig:2}(a) as $R_0$ decreases and $m$ increases, the value of $\theta^*$ decrease and $I_{h2}^*$ and $I_v^*$ increase. This is because parameter ranges corresponding to  $\alpha_{v}>1$ in Figure \ref{fig:2}(c) indicate vector feeding preference for $h1$ in the left column. Since 
           $I_{h1}^*$ increases as $m$ increases, the secondary infections result in higher $I_{h2}^*$ and $I_{v}^*$ in Figure \ref{fig:2}(g) and Figure \ref{fig:2}(i), respectively. 
           In Figure \ref{fig:2}(g), $I_{h2}^*$ is higher in parameter region for $\theta^*=0$ whereas in Figure \ref{fig:2}(h)  $I_{h2}^*$ higher in region corresponding to $\theta^*=1$. This is because in Figure \ref{fig:2}(d), $\alpha_{v}<1$ hence vector prefer $h2$ in right column. When $h1$ is unavailable to feed due to complete ITN adherence, vector bites for $h2$ increase $I_{h2}^*$. However, in Figure \ref{fig:2}(j), the $I_{v}^*$ is highest when $\theta^*=0$.  Hence, we notice that when vector prefer $h2$, the prevalence can be the lowest among $h2$ and the vector population when ITN users and non-users coexist. 
           This Figure illustrates that when the innate vector preference for $h2$ is low, increasing ITN coverage among $h1$ by lowering the ITN cost can reduce overall prevalence, and reducing $R_0$ when the cost of ITN is high may not reduce overall prevalence. However, when the innate vector preference for $h2$ is high, reducing $R_0$ by increasing $h1$ recovery rate can reduce the prevalence among $h1$, but the prevalence among $h2$ might not reduce even with full ITN coverage among $h1$.
	
     
	\subsection{Saddle node bifurcation with respect to parameter $\mu_1$ }

Two endemic equilibrium points can collide and disappear when parameters cross a certain threshold. Figure \ref{fig:3} illustrates the bifurcation diagram for the model system \eqref{eqn2} undergoing saddle-node bifurcation with respect to recovery rate of $h1$, $\mu_1$. We observe that for a low value of $\mu_1$, only $E_1^*$ and $E_2^*$ exist. The equilibrium $E_1^*$ is always unstable, and  $E_2^*$ is stable for values of $\mu_1$ up to $\mu_1^* \approx 0.3435$. At $\mu_1^*$,  model system \eqref{eqn2} exhibits saddle-node bifurcation as $E_2^*$ becomes unstable and a stable interior equilibrium $E^*$ emerges for values of $\mu_1 >\mu_1^*$. We numerically validate the existence of the interior endemic equilibrium point $E^*$ for $\mu_1 >\mu_1^*$ in Figure \ref{fig:7}. It shows the existence of root $\theta^*\in (0,1)$ for equation \eqref{eqntheta:01}, for $\mu_1=0.42$ and other parameter values as in Figure \ref{fig:3}. 
  
	 In Figure \ref{fig:3}(a) we observe that as $\mu_1$ decreases, the value of $I_{h1}^*,\; I_{h2}^*$, and $I_{v}^*$ increase slowly and at $\mu_1^* \approx 0.3435$, the stable $E^*$ disappears, and an unstable and stable equilibrium point emerge via a saddle-node bifurcation.  Moreover, as the value of $\mu_1$ decreases beyond the critical value of $\mu_1^*$, the value of stable point $I_{h1}^*,\; I_{h2}^*$, and $I_{v}^*$ rise rapidly. Such a high disease transmission rate can make managing VBD challenging, as it strains healthcare services. This underscores the importance of early diagnosis and prompt treatment to improve the recovery rate of $h1$ above the critical threshold $\mu_1^*$. 
	
	\subsection{Hopf bifurcation behavior with respect to parameter $m$ }
 To investigate how the dynamics of the coevolutionary feedback game develop periodic outbreaks, we vary parameter $m$ and plot the bifurcation diagram in Figure \ref{fig:4}. For low values of $m$, the interior equilibrium is stable as high ITN coverage reduces prevalence. As $m$ increases, the value of $\theta^*$ from the stable interior equilibrium $E^*$ decreases, which increases the value of $I_{h1}^*$,$I_{h2}^*$, and $I_{v}^*$. Further, $E^*$ loses its stability at $m\approx 10.38$, via a Hopf bifurcation, which causes the previously stable solution to oscillate periodically. The model system \eqref{eqn2} admits oscillatory solutions for values of $m$ within the range of $(10.83,21.33)$. It indicates the possibility of periodic outbreaks, caused by oscillatory behavior in using ITN protection among $h1$ individuals when $m$ lies within the range of $(10.83,21.33)$. High $I_{h1}^*$ increases $\theta^*$, and high $\theta^*$ results in a reduced VBD risk. The resulting negative feedback loop reduces ITN use among $h1$ individuals, which causes prevalence to rise again. This is because the perceived cost of ITN and  VBD risk are comparable and the tradeoff between them oscillates for intermediate values of $m$ which lie within the range of $(10.83,21.33)$.  Moreover, for large values of $m$, the perceived cost of ITN becomes larger compared to its benefit, hence ITN coverage declines, which causes the prevalence to rise. At $m\approx 21.33$ the model system \eqref{eqn2} undergoes Hopf bifurcation, and the oscillations die out as the interior equilibrium $E^*$ regains its stability.
  	
		\subsection{Effect of  $\mathcal{E}_{h1}$ on bifurcation structure}
	
 The unified framework for the occurrence of saddle-node and Hopf bifurcation in the model system \eqref{eqn2} as the encounter rate $\mathcal{E}_{h1}$ varies, is presented in Figure \ref{fig:6}.
At $\mathcal{E}_{h1}^* \approx 0.932$, the system undergoes a saddle-node bifurcation, where the stable interior equilibrium $E^*$ splits and gives rise to a pair of stable and unstable equilibria. When $\mathcal{E}_{h1}$ is below the critical value $\mathcal{E}_{h1}^*$, the prevalence level of stable equilibrium point $I_{h1}^*, I_{h2}^*$ and $I_{v}^*$ increase. However, for $\mathcal{E}_{h1} > \mathcal{E}_{h1}^*$, the prevalence values decreases slightly in $h2$ while $I_{h1}^*$ and $I_v^*$ remain nearly constant. As $\mathcal{E}_{h1}$ is increased further, the stable interior equilibrium point loses its stability at $\mathcal{E}_{h1} \approx 1.863$ and gives rise to a periodic solution via Hopf bifurcation. As we increase $\mathcal{E}_{h1}$ above the value of Hopf bifurcation point, the amplitude of oscillations increases with an increase in $h1$ encounter rate. Hence, high encounters with host $h_1$ {\em can} qualitatively alter the system dynamics. While moderate changes in $\mathcal{E}_{h1}$ only affect equilibrium prevalence levels, sufficiently large values can generate and amplify temporal fluctuations in disease prevalence. Such oscillatory dynamics hinder effective disease control by producing recurrent outbreaks with self-sustained fluctuations in both host and vector prevalence.

\section{Conclusion}\label{conclu}

In this work, we developed and analyzed a coevolutionary game-theoretic model for a vector-borne disease (VBD) and ITN use involving two host species and a vector population. An important feature of our model is that we derive vector feeding preference as a function of ITN coverage, allowing vector biting behavior to emerge directly from host protective behavior. By coupling this adaptive vector response with an evolutionary game describing ITN use by the preferred host, the model captures the feedback between disease prevalence, vector behavior, and host decision-making regarding ITN use, and reveals how their interactions shape disease outcomes.

Model analysis established threshold conditions for disease extinction and persistence through the reproduction numbers $R_0$ and $R_c$, and characterized the existence and stability of disease-free and endemic equilibria. Numerical simulations revealed that high ITN coverage can shift vector feeding preference towards $h_2$, which increases vector density and enlarges the parameter region in which complete ITN compliance remains stable, even for relatively high perceived costs of ITN use. However, while increased ITN coverage can substantially reduce prevalence in the protected host $h_1$, transmission may persist at relatively high levels in $h_2$ as vector feeding shifts toward the unprotected alternative host population.

The coupled disease-behavior system exhibits rich nonlinear dynamics. Using Sotomayor's theorem, conditions for the occurrence of saddle-node bifurcation were established and subsequently verified numerically. The bifurcation analysis showed the emergence of a stable endemic equilibrium from the bifurcation point and revealed that, beyond the critical recovery rate $\mu_1^*$, further increases in the recovery rate of $h_1$ lead to only marginal reductions in disease prevalence. This result highlights that the interventions focused solely on enhancing the host recovery rate may have a limited impact beyond the critical threshold.

The model was also shown to undergo Hopf bifurcation, giving rise to sustained periodic oscillations in both disease prevalence and ITN coverage. Numerical results demonstrated that increasing the perceived cost of ITN use can destabilize an otherwise stable endemic equilibrium and generate recurrent disease outbreaks through oscillatory behavioral feedbacks. Furthermore, variations in the encounter rate of $h_1$ were found to induce both saddle-node and Hopf bifurcations in the system. Low encounter rates promote higher disease prevalence through critical transitions associated with saddle-node bifurcation, whereas sufficiently high encounter rates generate large-amplitude oscillations in prevalence and ITN coverage. Our findings highlight that the coupling between interventions and behavioral factors can drive qualitative transitions between distinct dynamical regimes. 

Overall, this study provides a first step toward understanding the coupled effects of adaptive host behavior and vector feeding adaptation on the control of vector-borne diseases. Our results demonstrate how these interacting feedbacks can give rise to critical transitions and sustained oscillations, revealing rich dynamical behaviour that cannot be captured by models that ignore behavioral adaptation. We have also shown that although increased ITN coverage reduces prevalence in the protected host, adaptive shifts in vector feeding behavior can maintain high prevalence in the alternative host, potentially undermining overall disease control. This finding may have practical importance in the implementation and assessment of long-term control strategies for VBD. Future work should investigate how seasonal and environmental variability affects these feedbacks, and determine the conditions under which such variability amplifies or dampens recurrent disease outbreaks.

\appendix
\section{Impact of ITN on vector death rate and vector preference} \label{detail_appen}
Following from Stone {\em et al.} \cite{stone2018evolution}, we find dependence of daily vector mortality rate in the presence of ITN coverage. Suppose the conditional probability that vector feeds on host $h2$ after locating it is given as $\varphi=\frac{\sigma_{h2} \mathcal{E}_{h2}}{\sigma_{h1} \mathcal{E}_{h1}+ \sigma_{h2} \mathcal{E}_{h2}}$. Then the probability that the vector survives feeding on $h1$ and $h2$ is $p_{h1}=1-\nu \varphi^{\chi}$ and $p_{h2}=1-\nu (1-\varphi)^{\chi}$. Here, parameter $\nu$ is the maximum mortality rate due to host defense, and $\chi$ is the inverse of trade-off strength.
	Let the probability of a vector locating a host be denoted by $F=1-e^{-(\sigma_{h1} \mathcal{E}_{h1}+ \sigma_{h2} \mathcal{E}_{h2})}.$ If the vector survives a single foraging bout with probability $p_f$, then the probability of successfully locating a host in a single foraging attempt is $F p_f$. Now, let $r$ be the probability that the vector enters ITN and is repelled, $s$ be the probability that the vector is not killed by insecticide and manages to successfully feed on host $h1$. Then, the probability of the vector dying upon contact with ITN is $1-r-s$. Hence we can find the probability of vector acquiring a meal in one feeding attempt by considering the possibility that it can bite $h2$ with probability $p_f F \varphi p_{h2}$ or by biting an unprotected $h1$ host with probability $p_f F(1-\varphi)(1-\theta) p_{h1}$ or by biting a protected $h1$ host with probability $p_f F(1-\varphi) \theta (1-r) s p_{h1}$. Following from \cite{le2007elaborated,stone2018evolution}, the probability of successfully feeding in the presence of ITN is
	$W=p_f F(\varphi p_{h2}+ (1-\varphi)(1-\theta) p_{h1}+(1-\varphi) \theta (1-r) p_{h1} s).$
	Now, if the probability that the vector repeats or returns for another foraging bout is
	$\rho=(1-F) p_f +F p_f (1-\varphi) \theta r.$
	Then, the expected time to feed in a population where ITNs are used is the baseline time multiplied by the number of attempts required to
	complete a feeding cycle, $1/(1-\rho)$. Thus, the time to complete one feeding cycle is $\frac{1}{c}=(\frac{\tau_f}{1-\rho}+\tau_2).$ 

During a feeding attempt in the presence of ITN, the vector can succeed, die due to insecticide on the net, or be repelled and try again. Therefore, the probability of the vector surviving the feeding stage in the presence of ITN is found as $s_n=W+W\rho+W \rho^2+\dots=\frac{W}{1-\rho}$. Now, if the probability of the vector surviving the resting stage is $s_r$, then the total probability that the vector survives the foraging and resting stage of the complete feeding cycle is \[P_{fc}=(s_n s_r)^c.\]
Given the survival probability, we find the death rate as 
\[d_3=- \log(P_{fc})=- c \log(s_n s_r). \] 
Since $0 \le s_n s_r \le 1$, using the Taylor series expansion, we get
\[ d_3 \le  c \Big( \frac{1}{s_n s_r} -1 \Big). \]
Hence, we consider the following
\[ d_3\approx  \frac{c}{2} \Big( \frac{1}{s_n s_r} -1 \Big).\]

Now, we present the derivation of vector feeding preference under ITN use by $h1$. Following \cite{shravani&sharma} we consider the following
\[ \alpha_v =\frac{\text{Proportion of feeds on } h1}{\text{Proportion of feeds on } h2 } \times \frac{N_{h2}}{N_{h1}}.\]
When ITN is used by $h1$, a larger proportion of bites are on $h2$. Moreover, let $q(\theta)$ be the overall probability that the vector feeding trial will end up biting on $h1$. This is possible if the vector feeds on a non-protected $h1$, or successfully bites a protected $h1$, or has feeding success on protected $h1$ after attempting multiple times \cite{le2007elaborated}.
 \[q(\theta)= (1-\varphi) (1-\theta + s \theta ) +\rho q(\theta)
\implies q(\theta)= \frac{  (1-\varphi) (1-\theta + s\theta ) }{1-\rho  }.\]
Since $s_n(\theta)= W/(1- \rho)$ denotes the probability of the vector surviving foraging in the presence of ITN. Following from \cite{le2007elaborated} we can get the proportion of feeds taken on $h1$ as $$ \frac{q(\theta)}{s_n(\theta)}=\frac{(1-\varphi ) (1-\theta + s\theta )}{W}.$$
Consequently, the vector preference is given as follows
\[ 	\alpha_v(\theta)=\frac{(1-\varphi) (1-\theta+ s\theta )N_{h2} }{ \varphi W(\theta) N_{h1}}. \]

Moreover, we see that $\frac{\partial \alpha_v}{\partial \theta}<0$ iff $(1-\varphi) p_{h1} s r<\varphi p_{h2}(1-s)$. This indicates that feeding preference for $h1$ will reduce with increasing ITN coverage only if the resulting reduction in the success of biting protected $h1$, $(1-s)$, and getting diverted to $h2$ is stronger than the residual success in feeding $h1$ ($rs$). 

Following from \cite{stone2018evolution}, throughout our study, we assume the values of the above parameters as $\tau_f=0.1,\tau_2 =2.5,\sigma_{h1}=1,f_v=5,p_f=0.95,s_{r}=0.95,s=0.1,r=0.6,\chi=1,\nu=0.5.$

\section{Next generation matrix method for finding basic reproduction number}\label{R0NGM} 

The matrix $J_{DFE}$ can be decomposed as $J_{DFE}=T+\Sigma$. Here, $T$ accounts for VBD transmission terms and $\Sigma$ accounts for internal transitions such as mortality and recovery from VBD. Then the next generation matrix is $NGM=  -T\Sigma^{-1}$, where
\begin{align*}
 T &=\begin{bmatrix}
		0 &           0& \frac{N_{h1}  \beta_{hv} \alpha_{v}^* c ^*}{(\alpha_{v}^* N_{h1} + N_{h2})}& 0 & 0\\
		0&     0 & \frac{N_{h2}  \beta_{hv} c^*}{(\alpha_{v}^* N_{h1} + N_{h2})}& 0& 0\\
		\frac{K(\theta^*) \beta_{vh} \alpha_{v}^* c^*}{(\alpha_{v}^* N_{h1} + N_{h2})}&   \frac{K(\theta^*) \beta_{vh}  c^*}{(\alpha_{v}^* N_{h1} + N_{h2})}& 0& 0 & 0\\
		0&    0& 0& 0& 0\\
		0&    0& 0&0&  0\\
	\end{bmatrix},\\
\Sigma^{-1}&=	\begin{bmatrix}
                 \frac{-1}{(\mu_1+d_1)}&    0&    0&                   0&                   0\\
          0& \frac{-1}{(\mu_2+d_2)}&    0&                   0&                   0\\
          0&    0& \frac{-1}{d_3}&                   0&                   0\\
          \frac{-C_1 p \theta^* (1-\theta^*) (f_v c^{\prime}-d_3^{\prime})K(\theta^*)}{(\mu_1+d_1) a_1}&    0&    0&   \frac{ p(1-2\theta^*)(C_2  K(\theta^*)-m)}{a_1} &  \frac{-(f_v c^{\prime}-d_3^{\prime})K(\theta^*)}{a_1}\\
          \frac{-C_1 p \theta^* (1-\theta^*) (b_3-d_3)}{(\mu_1+d_1) a_1}&    0&    0&  \frac{-p\theta^*(1-\theta^*)C_2}{a_1} &   \frac{-(b_3-d_3)}{a_1}
	\end{bmatrix},\\
  \text{where, } & a_1=-(b_3-d_3) p(1-2\theta^*)(C_2  M^*-m)-p\theta^*(1-\theta^*)C_2 (f_v c^{\prime}-d_3^{\prime})M^*.
\end{align*}
We get the following 
\begin{align*}
NGM &=\begin{bmatrix}
   0&      0& \frac{N_{h1}  \beta_{hv} \alpha_{v}^* c ^*}{(\alpha_{v}^* N_{h1} + N_{h2})d_3^*}& 0& 0\\
0&      0& \frac{N_{h2}  \beta_{hv}  c^*}{(\alpha_{v}^* N_{h1} + N_{h2})d_3^*}& 0& 0\\
     \frac{K(\theta^*) \beta_{vh} \alpha_{v}^* c^*}{(\alpha_{v}^* N_{h1} + N_{h2}) (\mu_1+d_1)}&    \frac{K(\theta^*) \beta_{vh}  c^*}{(\alpha_{v}^* N_{h1} + N_{h2}) (\mu_2+d_2)}&      0& 0& 0\\
0&      0&      0& 0& 0\\
0&      0&      0& 0& 0\\
\end{bmatrix}.
\end{align*}
The dominant eigenvalue of NGM is $\Lambda=\sqrt{ \frac{\beta_{hv}\beta_{vh}c^{2}(\theta^*) K(\theta^*)}{(\alpha_{v}(\theta^*) N_{h1}+N_{h2})^2 d_3(\theta^*)}\left(\frac{\alpha_{v}^{2}(\theta^*) N_{h1}}{(\mu_{1} + d_1 )} +\frac{N_{h2}}{(\mu_2+d_{2})}\right)}.$ Hence DFE is stable if $$R_c(\theta^*)=\frac{\beta_{hv}\beta_{vh}c^{2}(\theta^*) K(\theta^*)}{(\alpha_{v}(\theta^*) N_{h1}+N_{h2})^2 d_3(\theta^*)}\left(\frac{\alpha_{v}^{2}(\theta^*) N_{h1}}{(\mu_{1} + d_1 )} +\frac{N_{h2}}{(\mu_2+d_{2})}\right)>1.$$

\section{Coefficients of characteristic equation for $J_{EE}$}\label{charceqn1}
		\begin{align}
		\begin{split}
			M_1&= - f_{11} - f_{22} - f_{33} - f_{44} - f_{55}-f_{66},\\ 
			M_2&= f_{11}f_{22} + (f_{11} + f_{22})f_{33}  + (f_{11}+ f_{22}  + f_{33})f_{44}  + ( f_{11}+f_{22} + f_{33}+ f_{44})f_{55} \\
			&+  ( f_{11}+f_{22} + f_{33}+ f_{44}+f_{55}) f_{66}- f_{12}\mu_1 - f_{34}f_{43}- f_{56}f_{65}-f_{14}f_{41}-f_{16}f_{61},\\ 
			M_3&=(f_{12} \mu_1+f_{14} f_{41}+f_{34} f_{43}-f_{11} f_{22}-(f_{11}+f_{22}) (f_{33}+f_{44})-f_{33} f_{44})(f_{55}+f_{66})  \\
			&  
			+ (f_{56} f_{65} -f_{55}  f_{66})(f_{44}+f_{33}+f_{22}+f_{11}) 
			+(f_{12} \mu_1 - f_{11} f_{22}   + f_{16} f_{61} )(f_{44}+f_{33}) \\
			&+(f_{34} f_{43}-f_{33} f_{44})(f_{11}+f_{22}) + f_{14} f_{41}(f_{33} +f_{22}) +  - f_{14} f_{46} f_{61} + f_{16} f_{61} (f_{55} +f_{22} ),\\
			M_4&=f_{16}f_{34}f_{43}f_{61}-f_{14}(f_{36}f_{43}+f_{45}f_{56})f_{61}+(f_{14}f_{41}+f_{34}f_{43})f_{56}f_{65}+f_{12}\mu_1(f_{34}f_{43}+f_{56}f_{65})
			\\&+f_{22}(f_{14}f_{46}f_{61}-f_{11}(f_{34}f_{43}+f_{56}f_{65}))+f_{33}(-f_{22}(f_{14}f_{41}+f_{16}f_{61})+f_{14}f_{46}f_{61}-(f_{11}+f_{22})f_{56}f_{65}) \\&+f_{44}((f_{11}f_{22}-f_{12}\mu_1)f_{33}-(f_{22}+f_{33})f_{16}f_{61}-(f_{11}+f_{22}+f_{33})f_{56}f_{65})+f_{55} f_{61}(f_{14}f_{46}\\&-f_{16}(f_{22}+f_{33}+f_{44}))+(f_{11}f_{22}+(f_{11}+f_{22}f_{33}+(f_{11}+f_{22}+f_{33)}f_{44}-f_{12}\mu_1-f_{14}f_{41}-\\&f_{34}f_{43})f_{55}f_{66}+(f_{55}+f_{66})((-f_{12}\mu_1+f_{11}f_{22})f_{33}-f_{14}(f_{22}+f_{33})f_{41}-(f_{11}+f_{22})f_{34}f_{43}\\&+f_{44}(-f_{12}\mu_1+f_{11}f_{22}+(f_{11}+f_{22})f_{33})),\\
			M_5&= (f_{55}+f_{66})(f_{14}f_{22}f_{33}f_{41}+(f_{12}\mu_1-f_{11}f_{22})(f_{33}f_{44}-f_{34}f_{43}))+(((f_{12}\mu_1-f_{11}f_{22})(f_{33}+f_{44})\\&+f_{14}f_{41}(f_{22}+f_{33})+(f_{11}+f_{22})(f_{34}f_{43}-f_{33}f_{44}))f_{66}+(f_{16}(f_{22}f_{33}-f_{34}f_{43}+(f_{22}+f_{33})f_{44})\\&+f_{14}(f_{36}f_{43}-(f_{22}+f_{33})f_{46}))f_{61})f_{55}+(f_{11}f_{22}-f_{12}\mu_1)f_{56}f_{65}(f_{33}+f_{44})+f_{44}f_{33}(f_{16}f_{22}f_{61}\\&+(f_{11}+f_{22})f_{56}f_{65})+f_{33}f_{14}(f_{56}(f_{61}f_{45}-f_{41}f_{65})-f_{22}f_{46}f_{61})\\&+f_{22}((f_{14}(f_{36}f_{43}+f_{56}(f_{45}-f_{41}f_{65}))-f_{16}f_{34}f_{43})f_{61})-(f_{11}+f_{22})f_{34}f_{43}f_{56}f_{65}\\
			M_6&=f_{14}f_{22}(-f_{36}f_{43}f_{55}f_{61}+f_{33}(f_{55}(f_{46}f_{61-}f_{41}f_{66})+f_{56}(f_{41}f_{65}-f_{45}f_{61})))\\&+(f_{34}f_{43}-f_{33}f_{44})(f_{16}f_{22}f_{55}f_{61}+(f_{11}f_{22}-f_{12}\mu_1)(f_{56}f_{65}-f_{55}f_{66})).
		\end{split}
	\end{align}
Note that at $\theta^*\in (0,1)$ we get $f_{66}=0$, which implies $M_1>0$.
    
\bibliography{sn-bibliography}

\section{Figures}\label{figureapp}
	\begin{figure}[t]
		\centering
	\includegraphics[width=0.8\textwidth]{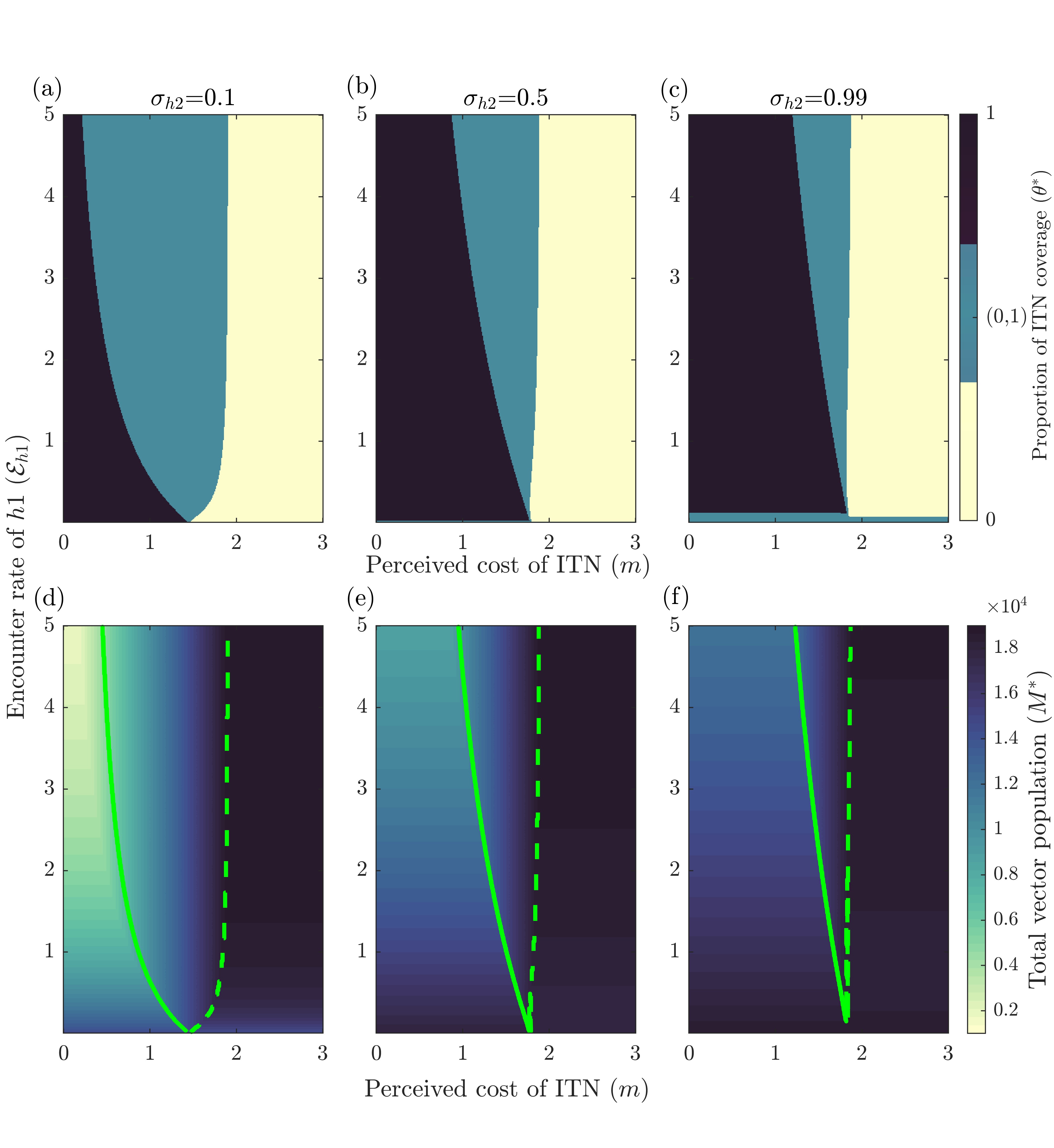}
		\caption{Impact of varying $\mathcal{E}_{h1}$ and $m$ on (a)-(c) strategy composition of $h1$ population, and (d)-(f) size of vector population, under different values of $\sigma_{h2}$. In (a)-(c) cream white and black colored regions denote the stability of DFE corresponding to $E_{03}$ and $E_{04}$. Light blue shaded region represents stability of  $E_{0 5}$, and region corresponding to it in (d)-(f) is denoted inside green line. That is, region below solid and above dashed curve bound the region corresponds to stability of $E_{04}$ and $E_{03}$, respectively. The value of $\alpha_v,c,d_3,b_3$ are varying as they are functions of $\mathcal{E}_{h1}$ and $\sigma_{h2}$. Other parameters are kept constant at $C_{1}=0.1,	C_{2}=0.0001,p=0.1, b_{31}=0.0001,\mathcal{E}_{h2}=1,\beta_{hv}=0.1,\beta_{vh}=0.2,\delta_{1}=0.01,\mu_2=0.01,d_{1}=0.1,d_{2}=0.5,\mu_{1}=0.5,N_{h1}=4000,N_{h2}=2000$. }
		\label{fig:1}
	\end{figure}

    \begin{figure}[htbp]
	\centering
	\includegraphics[width=1\textwidth]{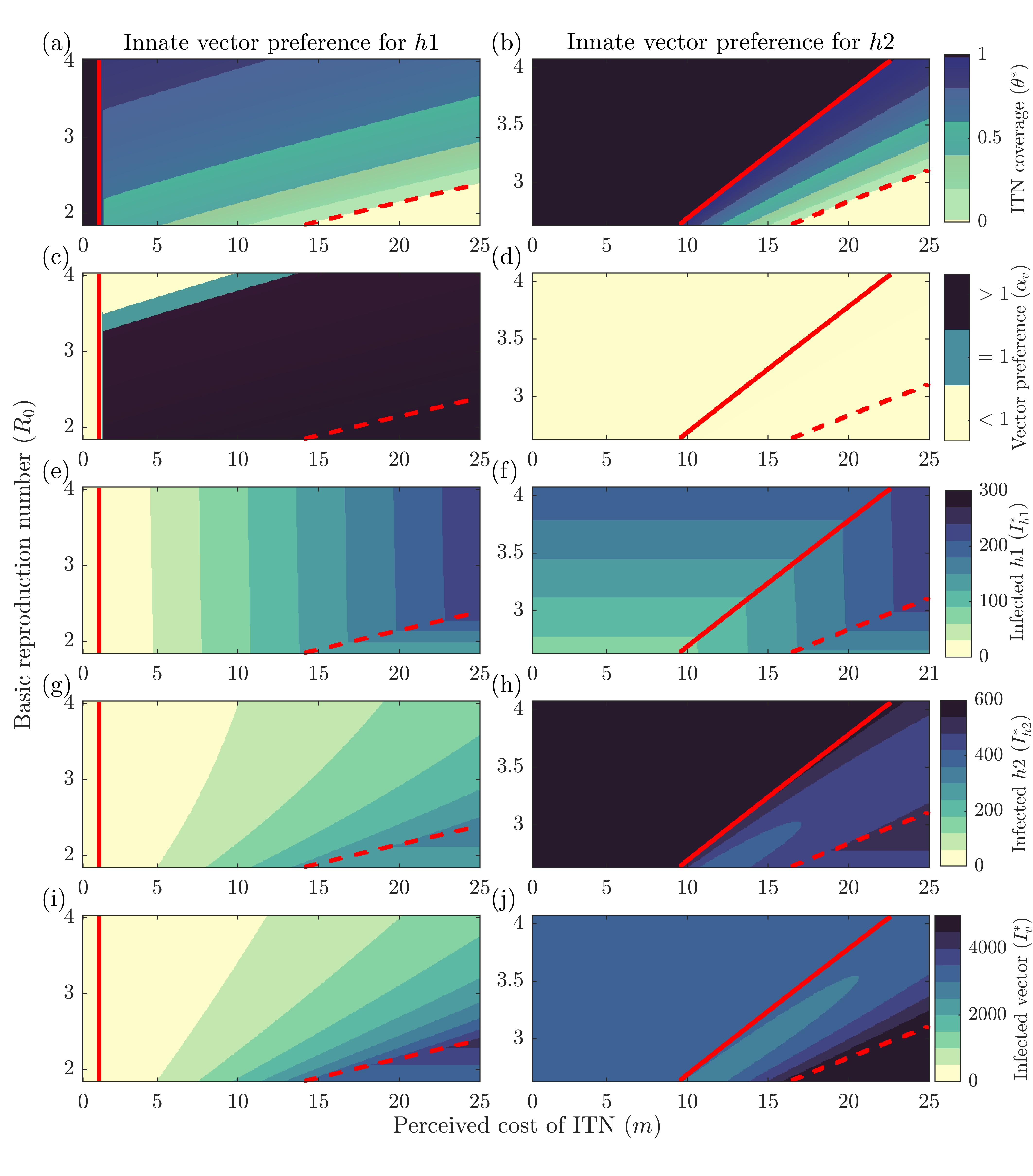}
	\caption{Relation between long-term behavior of model solutions and $R_0$ values for different values of $m$. Here, $R_0$ is calculated by varying $\mu_1$ from  0.5 to 2 and keeping $\theta=0$. 
    Panels demonstrate the following: 
    (a)-(b) equilibrium proportion of ITN coverage in $h1$, (c)-(d) equilibrium values of vector preference, 
    (e)-(f) equilibrium values of infected $h1$, 
    (g)-(h) equilibrium values of infected $h2$ and 
    (i)-(j) equilibrium values of infected vector.  Value of $\sigma_{h2}$ is $0.4$ and $0.7$ for figures corresponding to left and right panels. Innate vector preference without any ITN intervention depends on value of $\alpha_v(\theta)$ at $\theta=0$. In the left and right panels, we have $\alpha_v(0)=1.4628>1$ and $\alpha_v(0)= 0.8093<1$, that is, the innate vector preference is for $h1$ and $h2$ in the left and right panels respectively.
    Region below solid curve and above dashed curve corresponds to $\theta^*=1$ and $\theta^*=0$, respectively; the intermediate region satisfies $0<\theta^*<1$.  Parameter values used are $\beta_{hv}=0.5,\beta_{vh}=0.5,\delta_{1}=0.1,\mu_2=0.1,d_{1}=0.1,d_{2}=0.2,N_{h1}=3000,N_{h2}=1000,\mathcal{E}_{h1}=1$.  Other parameters were as in Figure \ref{fig:1}.}
	\label{fig:2}
\end{figure}

	\begin{figure}[htbp]
		\centering
		\includegraphics[width=0.8\textwidth]{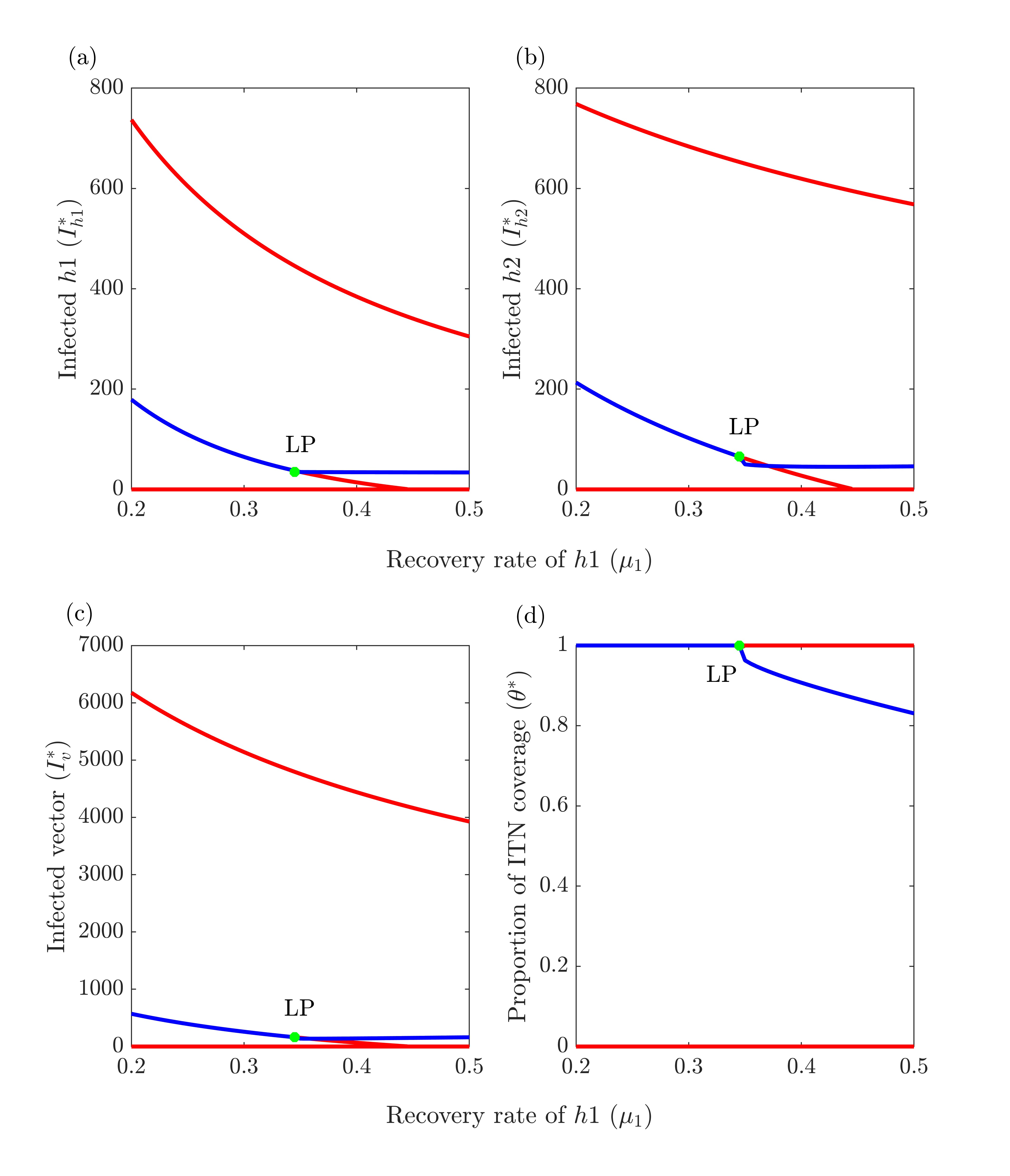}	
		\caption{Bifurcation diagram with respect to recovery rate of $h1$ indicating occurrence of saddle-node bifurcation at LP ($\mu_1 \approx 0.3435$).  Blue and red lines denote the stable and unstable values of equilibrium. The parameter values used were $m=5,\sigma_{h2}=0.5, \mathcal{E}_{h1}=1,\beta_{hv}=0.5,\beta_{vh}=0.5,\delta_{1}=0.01,\mu_2=0.01,d_{1}=0.05,d_{2}=0.2,N_{h1}=4000,N_{h2}=2000$ and rest parameter values are as in Figure \ref{fig:1}.}
		\label{fig:3}
	\end{figure}

	\begin{figure}[htbp]
		\centering
		\includegraphics[width=0.8\textwidth]{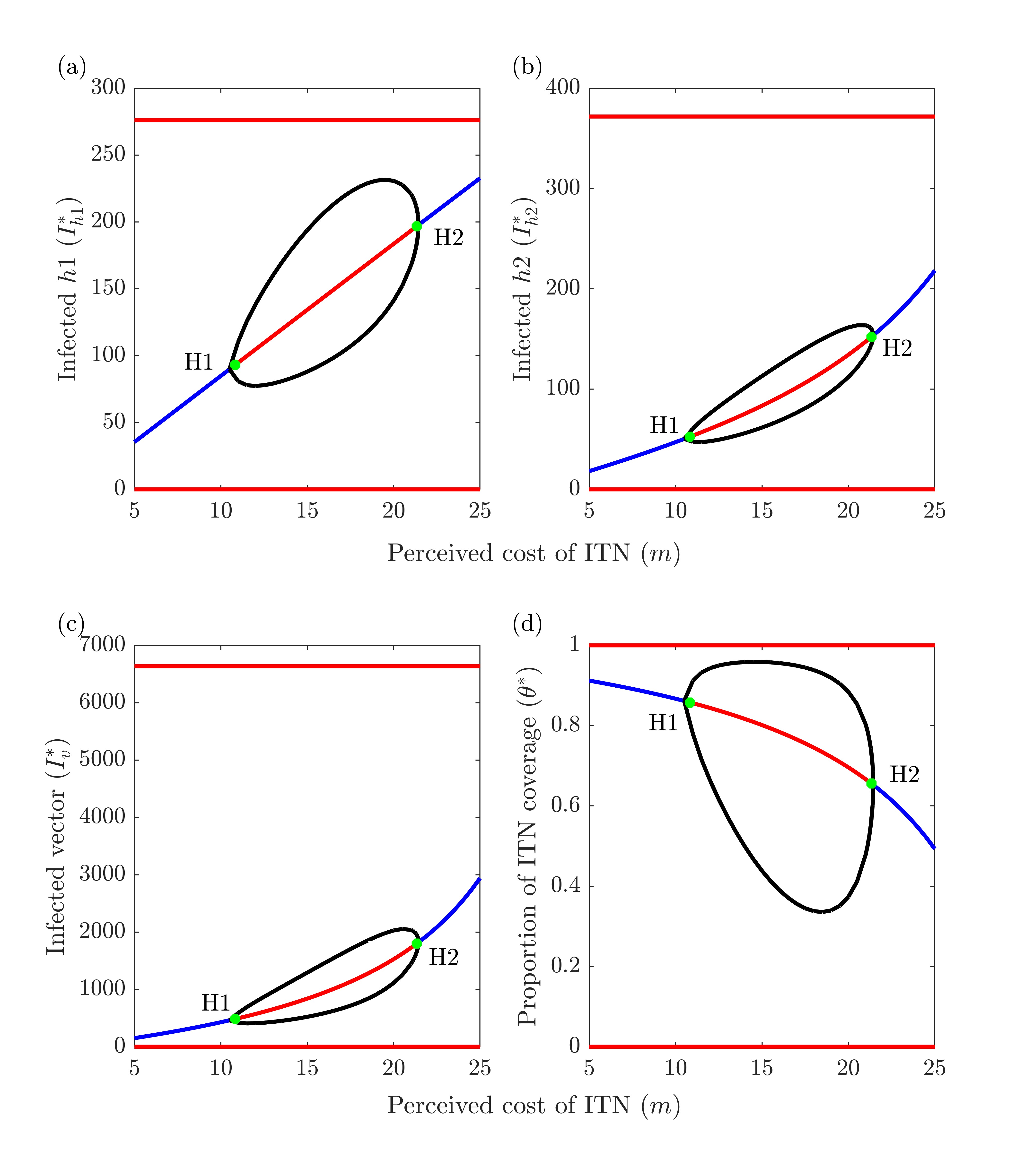} 
		\caption{ Bifurcation diagram with respect to cost of ITN ($m$). Hopf bifurcation is detected at H1 ($m \approx 10.83$) and H2 ($m \approx 21.33 $). The first lyapunov exponent of H1 and H2 are $-2.41 \times 10^{-5}$ and $5.15 \times 10^{-5}$.  Blue and red lines indicate the stable and unstable values of equilibrium.  Black line denotes the maximum and minimum values of periodic solutions. Parameter values used are $\sigma_{h2}=0.3,\beta_{hv}=0.5,\beta_{vh}=0.8,\delta_{1}=0.01,\mu_2=0.01,d_{1}=0.05,d_{2}=0.3,\mu_{1}=0.5,N_{h1}=3000,N_{h2}=1000,\mathcal{E}_{h1}=1$ and other parameters as in Figure \ref{fig:1}.}
		\label{fig:4}
	\end{figure}

	\begin{figure}[htbp]
		\centering 
		\includegraphics[width=0.8\textwidth]{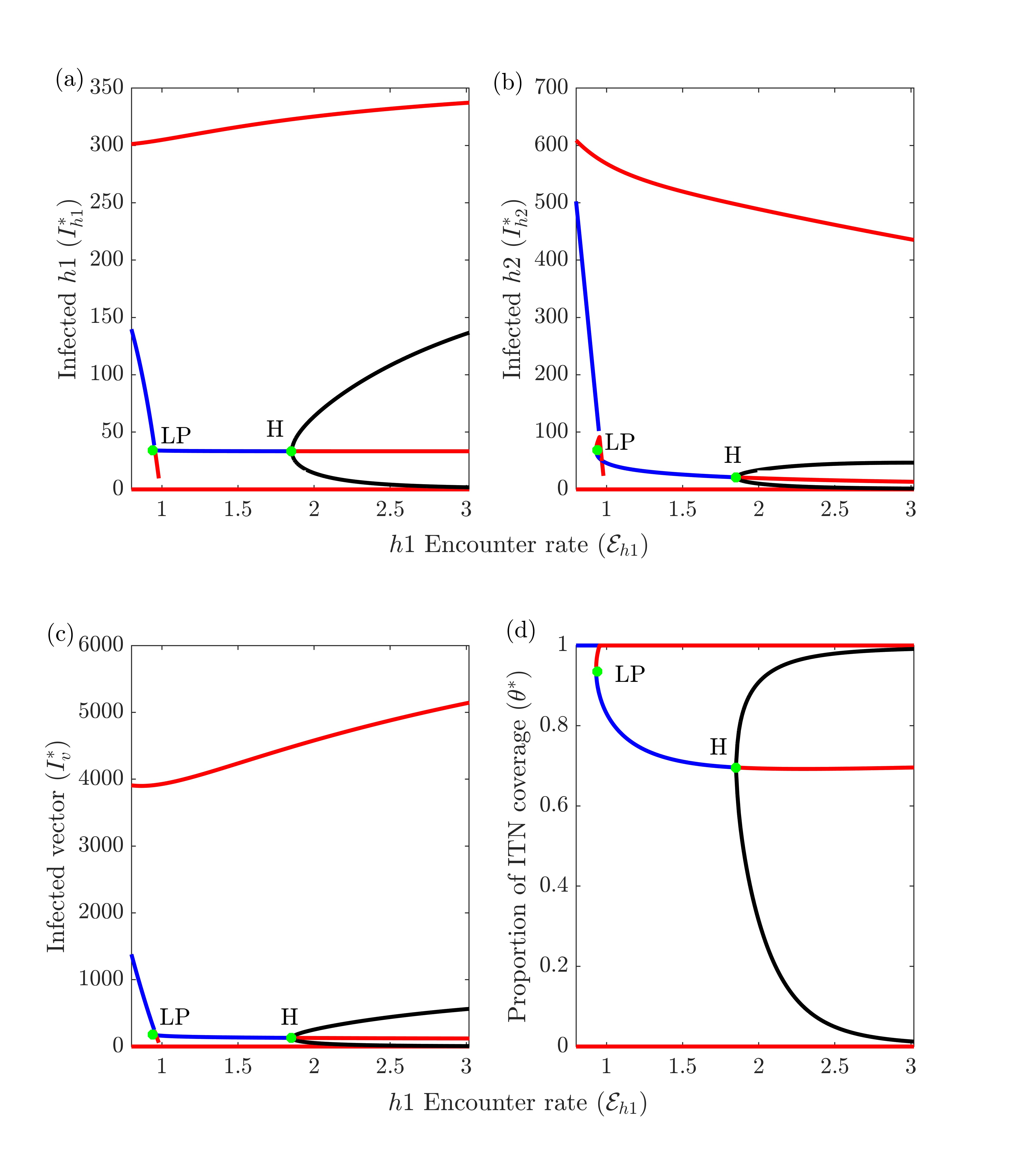}
		\caption{ Bifurcation diagram showing the endemic equilibrium points for system by varying $\mathcal{E}_{h1}$. Saddle-node and Hopf  bifurcation detected at LP ($\mathcal{E}_{h1} \approx 0.932$) and H 
			($\mathcal{E}_{h1} \approx 1.863$).  Blue and red lines represent the stable and unstable values of equilibrium.  Black line denotes the maximum and minimum values of periodic solutions.  Parameters used are $m=5,\sigma_{h2}=0.5,\beta_{hv}=0.5,\beta_{vh}=0.5,\delta_{1}=0.01,\mu_2=0.01,d_{1}=0.05,d_{2}=0.2,\mu_{1}=0.5,N_{h1}=4000,N_{h2}=2000$ and other parameters as in Figure \ref{fig:1}.}
	\label{fig:6}
\end{figure}

	 	\begin{figure}[htbp]
		\centering 
		\includegraphics[width=0.6\textwidth]{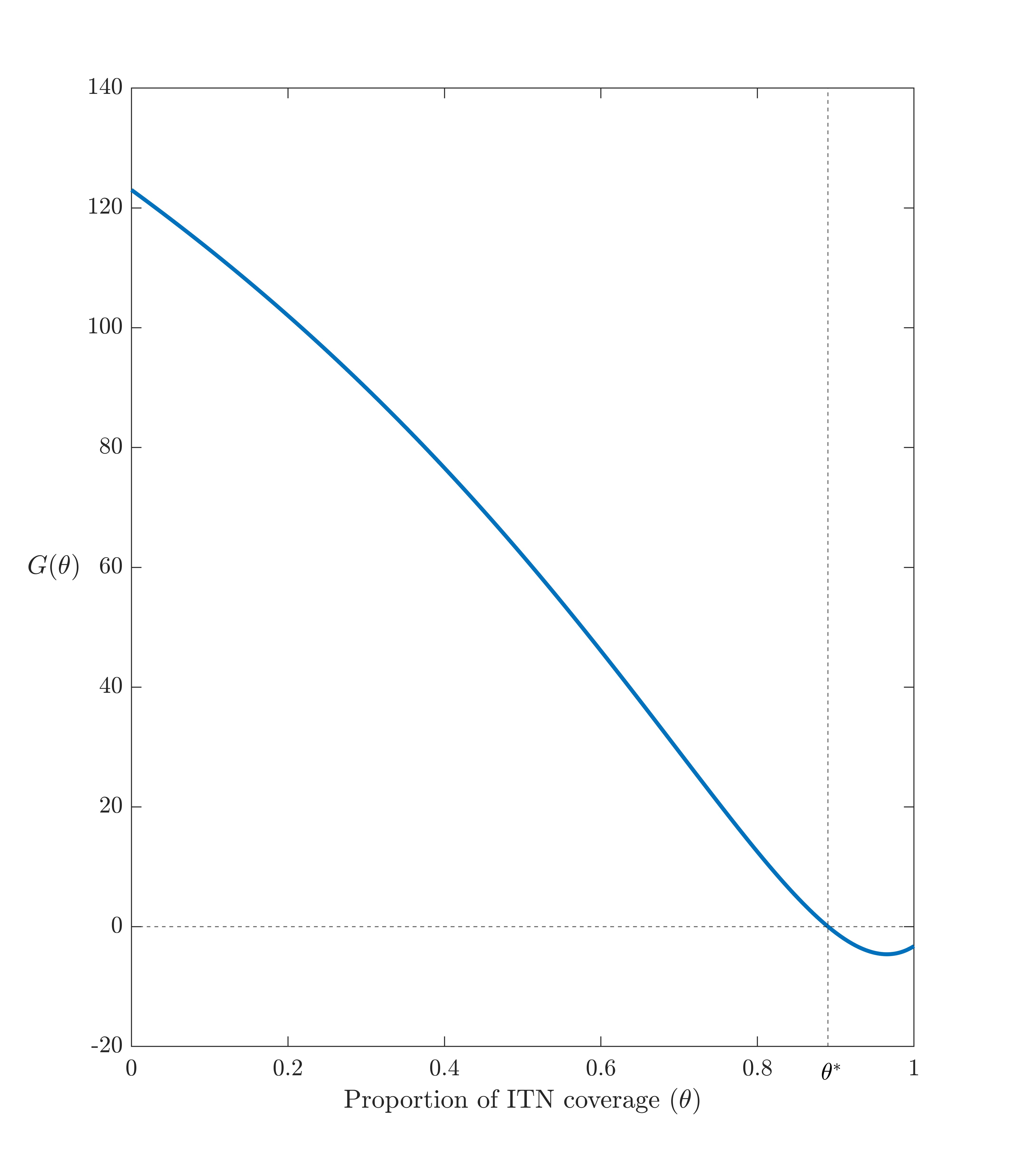}
		\caption{Existence of  positive real root for \eqref{eqntheta:01}. The value of $G(\theta)=0$ at $\theta^*$=0.89,  with $\mu_1=0.42$ and other parameter values same as in Figure \ref{fig:3}.}
	\label{fig:7}
\end{figure}

\end{document}